\documentclass[12pt]{article}

\usepackage{amssymb}
\usepackage{amsmath}
\usepackage{amsthm}
\usepackage{graphicx}
\usepackage{caption}
\usepackage{natbib}
\usepackage[colorlinks = true, citecolor = blue]{hyperref}
\usepackage[svgnames]{xcolor}
\usepackage[ruled,vlined,boxed]{algorithm2e}
\usepackage{stix}
\usepackage{pdflscape}
\usepackage{setspace}

\newcommand{\dcal}{\mathcal{D}}

\newcommand{\rcal}{\mathcal{R}}

\newcommand{\xcal}{\mathcal{X}}

\newcommand{\RR}{\mathbb{R}}

\newcommand{\BA}{\mathbf{A}}
\newcommand{\BB}{\mathbf{B}}
\newcommand{\BC}{\mathbf{C}}
\newcommand{\BD}{\mathbf{D}}

\newcommand{\BG}{\mathbf{G}}

\newcommand{\BI}{\mathbf{I}}

\newcommand{\BN}{\mathbf{N}}
\newcommand{\BO}{\mathbf{O}}
\newcommand{\BP}{\mathbf{P}}

\newcommand{\BT}{\mathbf{T}}

\newcommand{\BV}{\mathbf{V}}
\newcommand{\BW}{\mathbf{W}}
\newcommand{\BX}{\mathbf{X}}
\newcommand{\BY}{\mathbf{Y}}

\newcommand{\Bb}{\mathbf{b}}

\newcommand{\Bu}{\mathbf{u}}
\newcommand{\Bv}{\mathbf{v}}

\newcommand{\Bx}{\mathbf{x}}
\newcommand{\By}{\mathbf{y}}
\newcommand{\Bz}{\mathbf{z}}

\newcommand{\Bmu}{\boldsymbol{\mu}}
\newcommand{\Btheta}{\boldsymbol{\theta}}
\newcommand{\Beta}{\boldsymbol{\eta}}
\newcommand{\Bepsilon}{\boldsymbol{\varepsilon}}

\newcommand{\BLambda}{\boldsymbol{\Lambda}}
\newcommand{\BSigma}{\boldsymbol{\Sigma}}
\newcommand{\BOmega}{\boldsymbol{\Omega}}
\newcommand{\BPsi}{\boldsymbol{\Psi}}
\newcommand{\BGamma}{\boldsymbol{\Gamma}}

\newcommand{\BPhi}{\boldsymbol{\Phi}}

\newcommand{\Btau}{\boldsymbol{\tau}}
\newcommand{\Bxi}{\boldsymbol{\xi}}

\newcommand{\Blambda}{\boldsymbol{\lambda}}

\newcommand{\Bzero}{\mathbf{0}}

\newcommand{\tr}{{\scriptscriptstyle\mathsf{T}}}

\newcommand{\iid}{\overset{\textrm{iid}}{\sim}}
\newcommand{\indep}{\overset{\textrm{indep}}{\sim}}

\newcommand{\dd}{\textrm{d}}
\newcommand{\com}{,\,}
\newcommand{\given}{\,|\,}

\newcommand{\cas}{\overset{a.s.}{\to}}

\newcommand{\latent}{\Bz}
\newcommand{\unscaled}{\Bx}
\newcommand{\factor}{\Beta}
\newcommand{\copulacov}{\BC}

\newtheorem{theorem}{Theorem}
\newtheorem{lemma}[theorem]{Lemma}

\def\spacingset#1{\renewcommand{\baselinestretch}%
{#1}\small\normalsize} \spacingset{1}

\spacingset{1.9} 

    \usepackage{paralist}
    \usepackage[compact]{titlesec}
   \titlespacing{\section}{0pt}{0ex}{0ex}
    \titlespacing{\subsection}{0pt}{0ex}{0ex}
    \titlespacing{\subsubsection}{0pt}{0ex}{0ex}

\newif\ifblind

\begin{document}


\title{\Large \vspace{-5mm} A dynamic copula model for probabilistic forecasting of non-Gaussian multivariate time series \vspace{-5mm}}

\ifblind

\else

\author{John Zito\thanks{
Assistant Research Professor, Department of Statistical Science, Duke University (\href{mailto:john.zito@duke.edu}{john.zito@duke.edu}). 
}
\and
Daniel R.\ Kowal\thanks{Corresponding author. Associate Professor, Department of Statistics and Data Science, Cornell University (\href{mailto:dan.kowal@cornell.edu}{dan.kowal@cornell.edu}). Kowal gratefully acknowledges financial support from the National Science Foundation under grant number SES-2214726.}
}

\fi

\date{}

\maketitle

\vspace{-10mm} 
\begin{abstract} 
    Multivariate time series (MTS) data often include a heterogeneous mix of non-Gaussian distributional features (asymmetry, multimodality, heavy tails) and data types (continuous and discrete variables). Traditional MTS methods based on convenient parametric distributions are typically ill-equipped to model this heterogeneity. Copula models provide an appealing alternative, but present significant obstacles for fully Bayesian inference and probabilistic forecasting. To overcome these challenges, we propose a novel and general strategy for posterior approximation in MTS copula  models and apply it to a Gaussian copula built from a dynamic factor model. This framework provides scalable, fully Bayesian inference for cross-sectional and serial  dependencies and nonparametrically learns heterogeneous marginal distributions. We validate this approach by establishing posterior consistency and confirm  excellent finite-sample performance even under model misspecification using simulated data. We apply our method to crime count and macroeconomic MTS data and find superior probabilistic forecasting performance compared to popular MTS models.   These results demonstrate that the proposed method is a versatile, general-purpose utility for probabilistic forecasting of MTS that works well across of range of applications with minimal user input.

\end{abstract}

\vspace{-3mm}

\section{Introduction}\label{sec:intro}

Multivariate time series (MTS) data arise in economics \citep{karlsson2013chapter}, finance \citep{coa2009chapter}, neuroscience \citep{czrv2024chapter}, and many other fields. In MTS analysis, we seek to model the joint dynamics of many time series and generate probabilistic forecasts of future data. Figures~\ref{fig:crime}~and~\ref{fig:macro} display two examples of  MTS datasets:  monthly crime case counts in New South Wales, Australia and  quarterly US macroeconomic aggregates (plots for all series are in the supplementary material). A signal feature of these data is that each variable possesses its own unique distributional features, distinct from the other variables, and often non-Gaussian. Amongst the count-valued crime variables, we observe heavy tails (blackmail) and multimodality (escape from custody). Amongst the continuous macro variables, we observe multimodality (funds rate, yield spread) and asymmetry (unemployment). Occasionally, MTS include both discrete and continuous variables (e.g., \citealp{dt2023bern}). Faced with such MTS data, the challenge becomes modeling the cross-sectional and serial dependencies among the variables while simultaneously capturing the  marginal distributional features and data types of each series.

\begin{figure}[h]
    \centering
    \includegraphics[width=\textwidth]{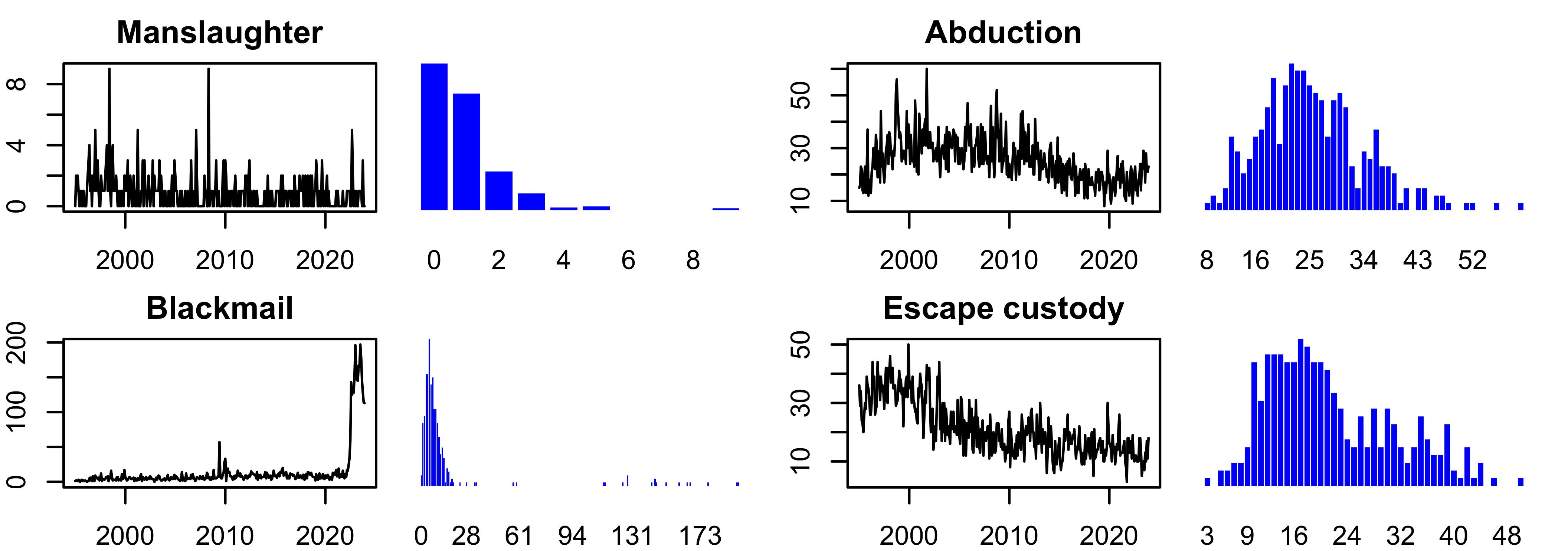}
    \caption{\small Monthly crime counts in New South Wales, Australia (January 1995 - December 2023). The full collection of variables is given in the supplementary material.}
    \label{fig:crime}
\end{figure}
\begin{figure}[h]
    \centering
    \includegraphics[width=\textwidth]{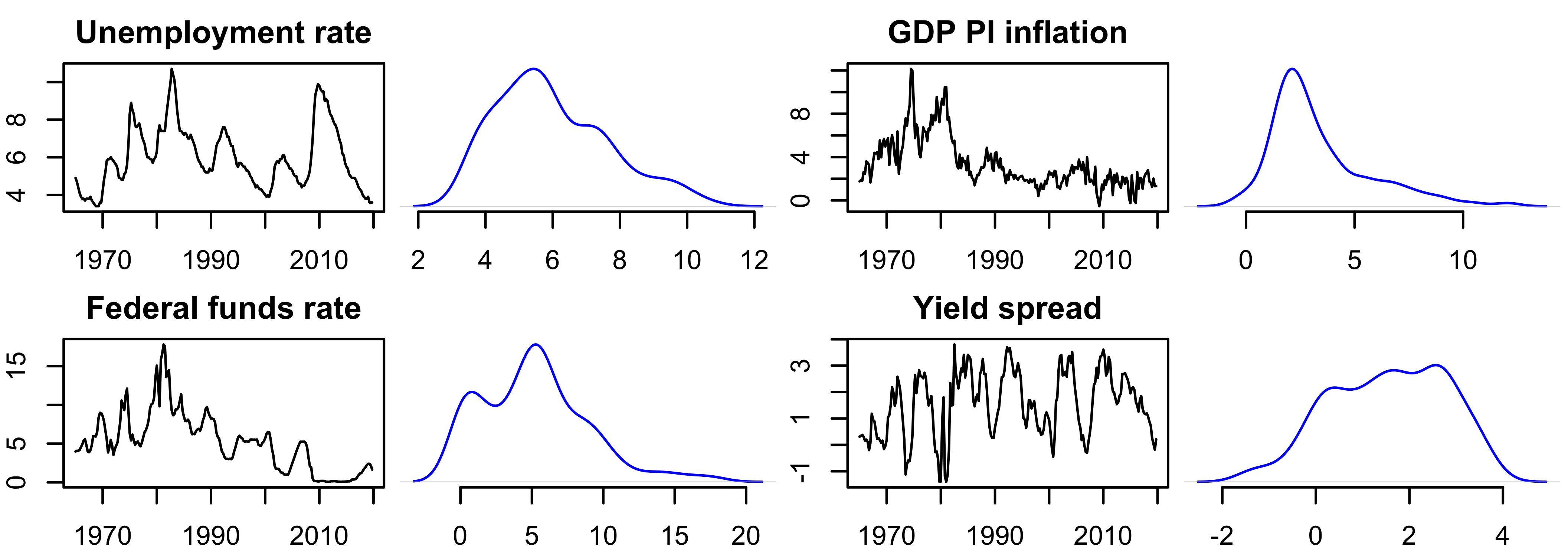}
    \caption{\small Quarterly US macroeconomic aggregates (1965Q1 - 2019Q4). The full collection of variables is given in the supplementary material.}
    \label{fig:macro}
\end{figure}

The rich literature on MTS analysis is strong on the subject of modeling cross-sectional and serial dependence, but perhaps weaker on the subject of flexibly capturing series-specific features within a multivariate framework. Most workhorse methods, such as vector autoregressive moving average (VARMA) models \citep{lutkepohl2006book} or dynamic linear models (DLMs; \citealp{wh1997book}), are based on convenient parametric families that make strict distributional assumptions:  each variable is Gaussian, or each variable is Student's $t$, etc. 
There is a modest literature on multivariate count-valued time series \citep{fokianos2024es}, including the class of dynamic generalized linear models \citep{whm1985ajasa} and the recent work of \cite{aps2018ba, aps2020jtsa}, but these methods make similarly sweeping parametric distributional assumptions. None of these methods apply for variables of mixed data types. As such, the usual methods are limited in the face of MTS that include a heterogeneous mix of distributional features and data types.

This weakness has motivated recent interest in copula models for MTS. Copula models provide a convenient framework for heterogeneous data because they effectively decouple the specification of the dependence structure between variables (and possibly across time) from the marginal behavior of each variable. Through this decoupling, each variable is assigned its own unrestricted marginal distribution 
and then these marginal distributions are ``stitched together'' in a principled way that captures the various dependencies but preserves the variable-specific behavior.


Copula methods for MTS fall into two main categories. 
In the first, a copula model is used to define the transition distribution in an autoregressive MTS model
\citep{cf2006joe, patton2013chapter, ct2015joe,dt2023bern}. 
In the second, which we pursue here, a large copula  defines the joint distribution across all variables and all time points  \citep{rps2012jma,smith2015ijf,sv2016jbes}. To be operational for MTS, these models must carefully consider cross-sectional and serial dependencies, forecasting recursions, and stationarity, among other things. These models are parametrized by two sets of unobservables: the copula parameters that model all (cross-sectional and serial) dependencies and the marginal distributions (or margins) that model the individual features of each variable. These constitute distinct ``levers'' that we can  manipulate in order to model a broad array of heterogeneous MTS data.

A primary goal of this paper is to generate \emph{probabilistic} forecasts for MTS, and in particular interval and density forecasts that convey reliable uncertainty quantification. We adopt a Bayesian approach, which is especially well-suited for this task via the posterior predictive distribution. However, posterior predictive inference requires careful accounting for the uncertainties in all model unobservables---which requires learning the joint posterior for all copula parameters and all marginal distributions. 
Unfortunately, existing Bayesian methods fall short in this difficult task. The vast majority of Bayesian approaches simply fix the margins at point estimates and then proceed with inference for the copula parameters as if the margins were known  (e.g., \citealp{ct2015joe,sv2016jbes}). Clearly, this strategy does not properly account for the uncertainty in the unknown margins.  \cite{smith2015ijf} is an exception that jointly infers the copula parameters and the margins, but requires parametric marginal distributions that are often inadequate for heterogeneous data. Furthermore, because of 1) the need to specify and tune each parametric margin and 2) the computational complexity of the posterior sampling algorithm, this approach can struggle to scale to larger cross-sections of variables \citep{ls2019jcgs}.


To overcome these limitations, we propose a new approach for Bayesian inference and probabilistic forecasting with MTS copula models. Our approach combines a rank-based posterior approximation for the copula parameters with a nonparametric procedure to infer the marginal distributions  (Section~\ref{sec:bayes}). Crucially, this approach minimizes user-intensive tuning: 
it does not require explicit specification of any marginal model, applies for mixed data types, and delivers posterior inference for the margins that adapts to the  particular features of each variable.
Taken together, this scheme provides an automated and computationally tractable method for joint inference on all model unknowns along with probabilistic forecasts. 
We validate our approach by providing model-trusting posterior consistency theory (Section~\ref{sec:theory}), and then show  that excellent performance is still achieved under model misspecification using simulated data (Section~\ref{sec:sims}). We apply our novel estimation scheme to a copula time series process built from a dynamic factor model, which is particularly well-suited for heterogeneous MTS with moderate to large cross-sections.  We compare this model to a basket of popular alternatives in a pair of probabilistic forecasting comparisons based on the crime count and macroeconomic datasets (Section~\ref{sec:applications}). 
Supplementary materials include proofs, additional simulation results, and details about the datasets. Replication files are available in the \texttt{GitHub} repository 
\ifblind
\texttt{[REDACTED]}.
\else
\href{https://github.com/johnczito/DynamicCopula}{\texttt{johnczito/DynamicCopula}}.
\fi

\section{Background}

\subsection{Copula models for multivariate time series}
We study MTS $\By_{1:T}=\{\By_1\com\By_2\com...\com \By_T\}$, where $\By_t=[y_{t,1}\,y_{t,2}\,\cdots y_{t,n}]^\tr$ is a cross-section of $n$ constituent univariate time series $\{y_{t,i}\}_{t=1}^T$. Each series  could have its own unique distributional features or data type, possibly quite different from all the others.  Our goal is to construct a probabilistic forecasting model for $\By_{1:T}$ that can capture this heterogeneity while simultaneously modeling both cross-sectional and serial dependencies.  In order to do this, we use  copula models.

A copula model is a joint probability distribution that is constructed to have an arbitrary set of marginal distributions. The ``inversion'' or ``implicit'' construction of a copula model \citep{smith2023es} has the following generative form:
\begin{align}
    \Bz_{1:T} &\sim p(\Bz_{1:T} \mid \Btheta)\label{eq:latent}\\
    y_{t,i}&=F_{t,i}^{-1}\circ G_{t,i}(z_{t,i})\label{eq:link}
\end{align}
where $\Bz_{1:T}=\{\Bz_1\com\Bz_2\com...\com \Bz_T\}$, $\Bz_t=[z_{t,1}\,z_{t,2}\,\cdots z_{t,n}]^\tr \in \mathbb{R}^n$ are latent variables with a joint distribution $p(\Bz_{1:T} \mid \Btheta)$ parametrized by $\Btheta$ and marginal distributions $G_{t,i}$. The construction in (\ref{eq:latent}\com \ref{eq:link}) is not specific to MTS, but rather defines a joint probability model  $p(\By_{1:T})$ that decouples the dependencies of $\By_{1:T}$ (modeled by  \eqref{eq:latent} and $\Btheta$) and the margins $y_{t,i} \sim F_{t,i}$. 

As a concrete example, the \emph{Gaussian copula} features a Gaussian latent distribution $\Bz_{1:T} \sim \text{N}_{Tn}(\Bzero\com\BC_{\theta})$ with  correlation $\copulacov_{\theta}$ indexed by $\Btheta$. Each margin for $\Bz_{1:T}$ is standard Gaussian,  $G_{t,i}=\Phi$. Despite its apparent simplicity, this model accommodates arbitrary marginal distributions, $y_{t,i}\sim F_{t,i}$, with dependencies among $\By_{1:T}$ determined by  $\BC_{\theta}$. Because of the  probability integral transform $F_{t,i}^{-1}\circ\Phi$,  the dependencies among $\By_{1:T}$ can be quite complex, including nonlinearities. 

To adapt (\ref{eq:latent}\com\ref{eq:link}) for MTS, we consider each level in turn. First, the latent model \eqref{eq:latent}  determines both serial and cross-sectional dependencies. 
We adopt a forward-looking, recursive characterization and instead specify a transition distribution $p(\Bz_{t}\given \Bz_{1:t-1}\com\Btheta)$ for the latent MTS. This framing grants access to a broad collection of MTS models, such as VARMAs, DLMs, dynamic factor models, etc. and operationalizes  (\ref{eq:latent}\com \ref{eq:link}) as a forecasting model for MTS. Specific choices are given in Section~\ref{sec:gauss_copula}, but our methods are general.

Second, we impose  some structure on the margins $\{F_{t,i}\}$, which in   \eqref{eq:link} are  variable- and time-specific. We require (\ref{eq:latent}\com \ref{eq:link}) to be strictly stationary:  set $F_{t,i}=F_i$ for all $t$ and specify $p(\Bz_{t}\given \Bz_{1:t-1}\com\Btheta)$ so that $\Bz_{1:T}$ is strictly stationary. Now, $F_i$ is the (unknown) marginal stationary distribution of $\{y_{t,i}\}_{t=1}^T$ and $G_{t,i}=G_i$ for all $t$. 
Then our generic template for specifying a MTS copula model  is 
\begin{align}
    \Bz_t&\sim p(\Bz_t\given\Bz_{1:t-1}\com\Btheta)\label{eq:ztrans}\\
    y_{t,i}&=F_i^{-1}\circ G_i(z_{t,i}).\label{eq:stationary_link}
\end{align}
Different choices of $p(\Bz_t\given\Bz_{1:t-1}\com\Btheta)$ produce different model configurations and propagate cross-sectional and serial dependencies to $\By_{1:T}$.  Variable-specific features are governed by the margins $\{F_i\}$, which can accommodate a heterogeneous mix of distributional features and data types. The model unknowns are  the copula parameters $\Btheta$ and the margins $\{F_i\}$; the latent stationary distributions $\{G_i\}$ are often known in closed-form and otherwise may be simulated from \eqref{eq:ztrans}. We evaluate our methods against real data that may violate strict stationarity in Section~\ref{sec:applications}. 



\subsection{Dynamic Gaussian factor copula models}\label{sec:gauss_copula}


In order to specify a latent transition distribution $p(\Bz_t\given\Bz_{1:t-1}\com\Btheta)$ with stationary dynamics, we turn to the rich class of linear, Gaussian time series models. These are well-studied, and bring with them a set of well-developed inferential tools that will ease the burden of computation later. To handle moderate to large cross-sections $n$, we specify a dynamic factor model (DFM) for \eqref{eq:ztrans} and refer to this model with \eqref{eq:stationary_link} as a \textit{dynamic Gaussian factor copula} (DGFC): 
\begin{align}
    \latent_t&=\BD_0^{-1/2}\unscaled_t, && \unscaled_t\given\Beta_t\indep\text{N}_n(\BLambda\Beta_t\com \BV)\label{eq:model_xz}\\
    \Beta_t&=\BG\Beta_{t-1}+\Bepsilon_t,&&\Bepsilon_t\iid\text{N}_k(\Bzero\com\BSigma). \label{eq:model_factor}
\end{align}
The DFM is applied to the latent variable $\unscaled_{t}$, which is rescaled by $\BD_0=\text{diag}\{\text{var}(\unscaled_{t})\}$ so that $\Bz_t$ in (\ref{eq:ztrans}\com \ref{eq:stationary_link}) is marginally (over $\Beta_t$) a  standard Gaussian process with $G_i = \Phi$. The factor model in \eqref{eq:model_xz} expresses  $\unscaled_{t} \in \mathbb{R}^n$ in terms of a static loadings matrix $\BLambda$ ($n \times k$) and dynamic factors $\Beta_t \in \mathbb{R}^k$  with $k \ll n$ and $\BV=\text{diag}(\Bv)$. These lower-dimensional factors $\Beta_t$ follow a vector autoregression (VAR) in (\ref{eq:model_factor}), which is the primary source of serial dependence in the model. The dynamics are determined by the transition matrix $\BG$, which we restrict to have eigenvalues in the unit circle for stationarity. In aggregate, the static parameters $\Btheta = \{\BG\com\BSigma\com\BLambda\com\Bv\}$ govern the cross-sectional and serial dependencies of $\Bz_{1:T}$ (and thus $\By_{1:T}$); a detailed accounting of these relationships is provided in the supplementary material. Default prior specifications are discussed in Section~\ref{sec:priors}. 


DFMs were introduced by \cite{geweke1977chapter} and are popular in macroeconometrics, where they are used to model large cross-sections of macroeconomic time series with a small number of latent factors \citep{sw2016chapter}. \cite{kn1998restat} introduced Bayesian analysis of these models, subsequently expanded upon by \cite{aw2000jbes} and \cite{bw2015jbes}, among many others. In our setting, DFMs are especially useful because they enable application of (\ref{eq:ztrans}\com \ref{eq:stationary_link}) for moderate to large cross-sections. In particular, DFMs provide both 1) model parsimony for MTS and 2) tractable properties that facilitate efficient posterior computations   (see Section~\ref{sec:priors}). However, our methods apply more broadly, including stationary Gaussian processes for $\Bz_{1:T}$ and  VAR copula models \citep{sv2016jbes,fhp2023joe}. 



\section{Bayesian inference and probabilistic forecasting}\label{sec:bayes}
Proceeding with the general MTS copula  model (\ref{eq:ztrans}\com \ref{eq:stationary_link}), there are two sets of unobservables: the static parameters $\Btheta$ that govern the cross-sectional and serial dependencies in $\By_{1:T}$ and the marginal distributions $F_{1:n}=\{F_1\com F_2\com ...\com F_n\}$ that govern the individual distributional features of each series. From a Bayesian perspective, the goal is to access the full, joint posterior distribution $p(F_{1:n}\com\Btheta\given\By_{1:T})$. This term is also the main ingredient for probabilistic forecasting via the multi-step  posterior predictive distribution $p(\By_{T+1:T+H}\given\By_{1:T})$, which we illustrate below. 

The primary challenge here is the margins $F_{1:n}$. Joint posterior inference for  $F_{1:n}$ and $\Btheta$ is sufficiently difficult that most Bayesian copula methods do not even attempt it, and instead fix the margins at point estimates 
\citep{ct2015joe, sv2016jbes}. This approach does not properly account for the uncertainty in the margins $F_{1:n}$, which can be substantial for short- to medium-length time series, and implicates  probabilistic forecasting from the posterior predictive distribution. Parametric margins offer a partial resolution \citep{smith2015ijf}, but limit the distributional flexibility and still require sophisticated and computationally intensive  algorithms.

In this environment, we seek a posterior approximation for $p(F_{1:n}\com \Btheta \given \By_{1:T})$ that 1) properly accounts for the joint uncertainties in $F_{1:n}$ and $\Btheta$, 2) allows for nonparametric modeling of each $F_{i}$ with minimal tuning, and 3) facilitates efficient computing. To provide an overview of the proposed approach, we begin with the following representation of the joint posterior under (\ref{eq:ztrans}\com\ref{eq:stationary_link}):
\begin{align}
    p(F_{1:n}\com\Btheta\given \By_{1:T})    &=    \int p(\latent_{1:T}\com F_{1:n}\com\Btheta\given\By_{1:T})  \, \dd \latent_{1:T} \label{eq:posterior} \\
     &= \int p(\latent_{1:T}\com\Btheta\given\By_{1:T}) \ p(F_{1:n}\given \latent_{1:T}\com\Btheta\com\By_{1:T})  \, \dd \latent_{1:T} \label{eq:m-c}
\end{align}
which applies a data augmentation in \eqref{eq:posterior} and a marginal-conditional decomposition in \eqref{eq:m-c}. 

Our approximation strategy sequentially considers the terms in the integrand of \eqref{eq:m-c}. First, $p(\latent_{1:T}\com\Btheta\given\By_{1:T})$ is unconditional on the margins $F_{1:n}$, which apparently requires an intractable marginalization. Instead, we approximate this distribution with a \emph{rank posterior} \citep{hoff2007aoas} that  conditions only on the ranks of $\{y_{t,i}\}_{t=1}^T$ for each series  (Section~\ref{sec:erl}). Crucially, rank posteriors are invariant to the margins $F_{1:n}$ and thus approximate $p(\latent_{1:T}\com\Btheta\given\By_{1:T})$ without the need to specify, estimate, or marginalize over $F_{1:n}$. 
However, full posterior inference for the margins remains a priority and is essential for probabilistic forecasting. Thus, the rank posterior alone is insufficient for our goals. This is resolved by our treatment of the second term, $p(F_{1:n}\given \latent_{1:T}\com\Btheta\com\By_{1:T})$: we show that the margins $F_{1:n}$ are essentially determined when given $\latent_{1:T}$ and $\By_{1:T}$ (Section~\ref{sec:ma}). Remarkably, this approach does not require any explicit model for $F_{1:n}$ and thus delivers automatic and nonparametric inference for these marginal stationary distributions. 

In conjunction, these approximations produce the pseudo-posterior
\begin{equation}\label{eq:pseudo}
        \tilde{p}(F_{1:n}\com\Btheta\given \By_{1:T})   = \int 
        \tilde{p}(\latent_{1:T}\com\Btheta\given\By_{1:T})\ 
        \tilde{p}(F_{1:n}\given \latent_{1:T}\com\Btheta\com\By_{1:T})
    \,
    \dd \latent_{1:T}
    .
\end{equation}
Subsequently, we will show that \eqref{eq:pseudo} is especially convenient and efficient to compute for many choices of (\ref{eq:ztrans}\com \ref{eq:stationary_link})  (Algorithm~\ref{alg:post})  and demonstrates favorable theoretical (Section~\ref{sec:theory}) and empirical (Section~\ref{sec:sims}) properties for joint inference on $F_{1:n}$ and $\Btheta$. 


Lastly, we build upon the approximation \eqref{eq:pseudo} to enable probabilistic  forecasting. Similar to (\ref{eq:posterior}\com \ref{eq:m-c}), the multi-step posterior predictive distribution can be expressed via data augmentation, now using both historical and $H$-step ahead latent data $\Bz_{1:T+H} = \{\Bz_{1:T}\com \Bz_{T+1:T+H}\}$: 
$p(\By_{T+1:T+H}\given \By_{1:T}) = \int p(\Bz_{1:T+H}\com F_{1:n}\com\Btheta\com\By_{T+1:T+H} \given \By_{1:T}) \,  \dd  \Bz_{1:T+H} \,
    \dd  F_{1:n} \, \dd \Btheta$. Again using marginal-conditional decompositions, the multi-step posterior predictive distribution can be represented conveniently as 
\begin{multline} \label{eq:pseudo-pred}
    p(\By_{T+1:T+H}\given \By_{1:T})  = \int 
    p(\By_{T+1:T+H} \given \Bz_{1:T+H} \com F_{1:n}\com\Btheta\com\By_{1:T})\  \times \\  p(\Bz_{T+1:T+H} \given \Bz_{1:T}\com F_{1:n}\com\Btheta \com \By_{1:T})
    \ p(\Bz_{1:T}\com F_{1:n}\com\Btheta\given \By_{1:T})
    \,  \dd  \Bz_{1:T+H} \dd  F_{1:n} \, \dd \Btheta.
    \end{multline}
First,  $p(\By_{T+1:T+H} \given \Bz_{1:T+H} \com F_{1:n}\com\Btheta\com\By_{1:T}) = p(\By_{T+1:T+H} \given \Bz_{T+1:T+H} \com F_{1:n})$  
is degenerate: given latent data and the margins, the (future) observations  are completely determined  by the probability integral transform \eqref{eq:stationary_link}. Second, $p(\Bz_{T+1:T+H} \given \Bz_{1:T}\com F_{1:n}\com\Btheta \com \By_{1:T}) = p(\Bz_{T+1:T+H} \given \Bz_{1:T} \com\Btheta)$ is simply the $H$-step forecasting distribution of the latent time series \eqref{eq:ztrans} with given parameters. Finally, the third term is precisely what we approximate in the integrand of \eqref{eq:pseudo}, i.e., $\tilde p(\Bz_{1:T}\com F_{1:n} \com \Btheta \given \By_{1:T})$. Crucially, access to \eqref{eq:pseudo-pred} requires minimal computing effort beyond the posterior approximation in   \eqref{eq:pseudo}  and fully leverages the forward-looking specification \eqref{eq:ztrans} of the latent model \eqref{eq:latent}; details are in Algorithm~\ref{alg:forecast}. 
With this, we can produce and evaluate the full range of point, interval, and density forecasts---which is not typical for MTS copula models.

\subsection{The rank posterior}\label{sec:erl}


The first task is to approximate $p(\latent_{1:T}\com\Btheta\given\By_{1:T})$, which notably does not depend on the margins $F_{1:n}$.  Rather than specifying a prior on  $F_{1:n}$ and attempting to integrate analytically (which would be intractable), we propose a rank-based posterior approximation. This strategy offers several key features: 1) it is invariant to the margins $F_{1:n}$, yet our approach  still provides joint inference for $F_{1:n}$ and $\Btheta$ (Section~\ref{sec:ma});  2) it is computationally efficient for many latent MTS models \eqref{eq:ztrans} using Gibbs sampling (see the supplementary material, Algorithm~\ref{alg:gibbs_sampler}); and 3) it offers favorable theoretical properties (Section~\ref{sec:theory}). 

To see why the rank posterior is especially useful for MTS copula models (\ref{eq:ztrans}\com \ref{eq:stationary_link}), notice that  $F_i^{-1}\circ G_i$   is nondecreasing, so   $z_{t,i}<z_{t',i}$ whenever $y_{t,i}<y_{t',i}$, for each series $i=1\com...\com n$. Crucially, this is true regardless of the margins $F_{1:n}$. The basic idea is to replace the observed data values with their relative ordering in our information set, which results in some information loss but removes the need to specify, estimate, or marginalize over $F_{1:n}$ at this stage.  More specifically, 
\eqref{eq:stationary_link} implies that $\latent_{1:T}$ must obey the same partial ordering as $\By_{1:T}$, meaning that  $\latent_{1:T}$ must reside in the set
  $  \dcal(\By_{1:T})=\{\latent_{1:T}\in\RR^{T\times n}:\max_k\{z_{k,i}:y_{k,i}<y_{t,i}\}<z_{t,i}<\min_{k}\{z_{k,i}:y_{t,i}<y_{k,i}\}\}.$
For continuous variables, this defines a total ordering, and the set $\dcal(\By_{1:T})$ encodes the same information as the ranks of the data $\By_{1:T}$. When  $y_{t,i}=y_{t',i}$ (e.g., for discrete variables), then $\dcal(\By_{1:T})$ defines only a partial ordering. 

Our approximation of $p(\latent_{1:T}\com\Btheta\given\By_{1:T})$ relies on the following decomposition, which we derive in the supplementary material (Appendix~\ref{app:rank}):
\begin{equation}
    p(\latent_{1:T}\com\Btheta\given\By_{1:T})
    = \underbrace{p\{\Bz_{1:T}\com\Btheta\given\Bz_{1:T}\in\dcal(\By_{1:T})\}}_{\text{rank posterior}}
\,
\times 
\,
    \frac{
    p\{\By_{1:T}\given\Bz_{1:T}\in\dcal(\By_{1:T})\com\Bz_{1:T}\com\Btheta\}
    }
    {
    p\{\By_{1:T}\given\Bz_{1:T}\in\dcal(\By_{1:T})\}
    }.\label{eq:rank}
\end{equation}
We use the   \textit{rank posterior} $ p\{\Bz_{1:T}\com\Btheta\given\Bz_{1:T}\in\dcal(\By_{1:T})\}$ as our approximation $\tilde{p}(\latent_{1:T}\com\Btheta\given\By_{1:T})$ in \eqref{eq:pseudo}.  The intuition, which we formalize later in Theorem~\ref{thm:doob1}, is that the ranks contain most of the relevant information for inference and thus the righthand term in \eqref{eq:rank} is negligible. Crucially, the rank posterior is invariant to the margins: this occurs because the rank \emph{likelihood} for the model unobservables satisfies $p\{\Bz_{1:T}\in\dcal(\By_{1:T}) \given \Btheta \com F_{1:n} \} = p\{\Bz_{1:T}\in\dcal(\By_{1:T}) \given \Btheta  \}$.
This obviates the need to marginalize over $F_{1:n}$ analytically, or to  even specify a model for $F_{1:n}$. Similar strategies have been used successfully with independent and identically distributed data \citep{hoff2007aoas,mdcl2013jasa}, but these previous approaches did not consider MTS data, inference for the margins $F_{1:n}$, or probabilistic forecasting.

To estimate the rank posterior for MTS copula models, we design a general Gibbs sampling algorithm. There are two main blocks:
\begin{compactenum}
    \item $p\{\Btheta\given \Bz_{1:T}\com \Bz_{1:T}\in\dcal(\By_{1:T})\} = p(\Btheta\given \Bz_{1:T})$ is exactly the posterior distribution under the MTS model \eqref{eq:ztrans} given data $\Bz_{1:T}$; and

    \item $p\{z_{t,i}\given \Btheta\com \Bz_{1:T}\setminus\{z_{t,i}\}\com\Bz_{1:T}\in\dcal(\By_{1:T})\} = p(z_{t,i}\given \Btheta\com \Bz_{1:T}\setminus\{z_{t,i}\})\ \mathbf{I}\{z_{t,i} \in (L_{t,i}, U_{t,i})\}$ for $t=1,\ldots,T$ and $i=1,\ldots,n$ is the full conditional distribution of \eqref{eq:ztrans} truncated to  an interval with endpoints $L_{t,i}=\max_{t'}\{z_{t',i}:y_{t',i}<y_{t,i}\}$ and $U_{t,i}=\min_{t'}\{z_{t',i}:y_{t',i}>y_{t,i}\}$. 
\end{compactenum}
Block 1  allows us to  apply existing algorithms for MTS models, including  VARMAs, DLMs, DFMs, etc. While these models may be unsatisfactory when applied directly to MTS with heterogeneous features (Figures~\ref{fig:crime}~and~\ref{fig:macro}), the computational tools for these models are immediately applicable for MTS copula models (\ref{eq:ztrans}\com \ref{eq:stationary_link}). For the DGFC with DFM (\ref{eq:model_xz}\com \ref{eq:model_factor}), we provide the sampling steps in Algorithm~\ref{alg:gibbs_sampler}. 
In Block 2, we further decompose $p\{\Bz_{1:T}\given \Btheta\com\Bz_{1:T}\in\dcal(\By_{1:T})\}$, which truncates the $Tn$-dimensional joint distribution in \eqref{eq:ztrans}, into univariate distributions truncated to an interval. This massively improves computational scalability in both $T$ and $n$, albeit at some expense to Monte Carlo efficiency. When \eqref{eq:ztrans} is a Gaussian time series model, the joint and full conditional distributions may be written  $\Bz_{1:T}\given\Btheta\sim\text{N}_{Tn}(\Bzero\com\BC_{\Btheta})$ and $z_{t,i}\given \Btheta\com \Bz_{1:T}\setminus\{z_{t,i}\}\sim\text{N}(\mu_{t,i}\com\sigma^2_{t,i})$, respectively, where the moments derive from standard conditioning results for Gaussian variables. Conveniently, this shows that Block 2 is  rather generic and broadly applicable across latent Gaussian time series models. Note that in general, computing $(\mu_{t,i}\com\sigma^2_{t,i})$ requires a large matrix inversion. However, our implementation of the DGFC circumvents this issue, namely because $z_{t,i}$ is conditionally independent of $\Bz_{1:T}\setminus\{z_{t,i}\}$ given $\Btheta$ and the factors $\Beta_{1:T}$ (see Algorithm~\ref{alg:gibbs_sampler}). 






\subsection{Inference for the marginal stationary distributions}\label{sec:ma}


The second task is to approximate $p(F_{1:n}\given \latent_{1:T}\com\Btheta\com\By_{1:T})$, which would complete the joint posterior distribution for all copula parameters $\Btheta$ and margins $F_{1:n}$ in \eqref{eq:pseudo}. At first glance, this appears infeasible: the approximation $\tilde p(\Bz_{1:T}\com  \Btheta \given \By_{1:T})$ from Section~\ref{sec:erl} used the rank posterior \eqref{eq:rank} and eschewed any specification, estimation, or marginalization over $F_{1:n}$. Yet perhaps surprisingly, we show that is both possible and rather straightforward to identify $p(F_{1:n}\given \latent_{1:T}\com\Btheta\com\By_{1:T})$---and to do so without any new modeling assumptions or specifications of $F_{1:n}$. 

For ease of exposition, we discuss a single $F_i$, and we momentarily drop the variable index $i$ for simplicity. Since the conditioning set includes the latent data $\Bz_{1:T}$, the parameters $\Btheta$, and the data $\By_{1:T}$, these are all treated as given in this section. 


The singular appearance of $F$ in the MTS copula model (\ref{eq:ztrans}\com\ref{eq:stationary_link}) occurs in the probability integral transform  $y_t=F^{-1}\circ G(z_t)$, which implies that $F(y_t)=G(z_t)$. Crucially, with our conditioning set, both $\By_{1:T}$ and $\Bz_{1:T}$ are given. Then $F$ is known \emph{exactly} at each $\{y_t\}_{t=1}^T$ (see Figure~\ref{fig:ma}, top). Put another way, $p\{F(x) \given \latent_{1:T}\com\Btheta\com\By_{1:T}\}$ is degenerate at $G(z_t)$ whenever $x=y_t$. For any unobserved point $x \not\in \{y_t\}_{t=1}^T$, the value of $F$ remains uncertain. Typically, this is where a prior specification for $F$ would contribute some information or smoothness. However, specifying a prior for the marginal distribution of each series is user-intensive even for moderate cross-sections and requires careful consideration of each variable's distributional features, support, and data type (e.g., real-valued, positive continuous, count, etc.). Thus, we pursue a simpler approach free of user input. 

Notice that $F$ must be bounded, non-decreasing, and pass through the pairs $\{(y_t\com G(z_t))\}_{t=1}^T$. Thus, we simply approximate  $p\{F(x) \given \latent_{1:T}\com\Btheta\com\By_{1:T}\}$ to be degenerate at a  distribution function $\tilde F$ that interpolates $\{(y_t\com G(z_t))\}_{t=1}^T$. As $T$ grows, the points along this curve will become dense and the approximation errors will dissipate.   
By default, we adopt the simplest interpolation strategy (Figure~\ref{fig:ma}, middle): a step function  that is zero for $x < \min\{y_{1:T}\}$, one for $x \ge \max\{y_{1:T}\}$, and otherwise ($\min\{y_{1:T}\}\leq x < \max\{y_{1:T}\}$) equal to
\begin{equation}\label{eq:ma}
    \tilde{F}(x)=\underset{t}{\max}\{G(z_t):y_t\leq x\}.
\end{equation}
We refer to (\ref{eq:ma}) as the \textit{margin adjustment}, following the terminology of \cite{fk2024jmlr}, who used a similar idea for imputation of independent and identically distributed data. When $F$ is continuous, each $y_t$ has a unique value, and so $\tilde{F}$ is a step function passing through the pairs $(y_t\com G(z_t))$ (Figure~\ref{fig:ma}, middle). When $F$ is discrete and each unique value of $y$ could have multiple values of $G(z)$ associated with it, (\ref{eq:ma}) takes the largest. 

This approximation is simple, computationally efficient, and well-behaved (Theorem~\ref{thm:ma}). Other interpolations may be preferred in some settings, such as a monotone spline that interpolates $\{(y_t\com G(z_t))\}_{t=1}^T$ to provide smoothness in place of a step function. However, we emphasize that the central feature is the interpolation; the behavior of $F(x)$ at unobserved points $x \not \in \{y_t\}_{t=1}^T$ becomes less and less important as $T$ grows.

\begin{figure}[h]
    \centering
    \includegraphics[width=.7\textwidth]{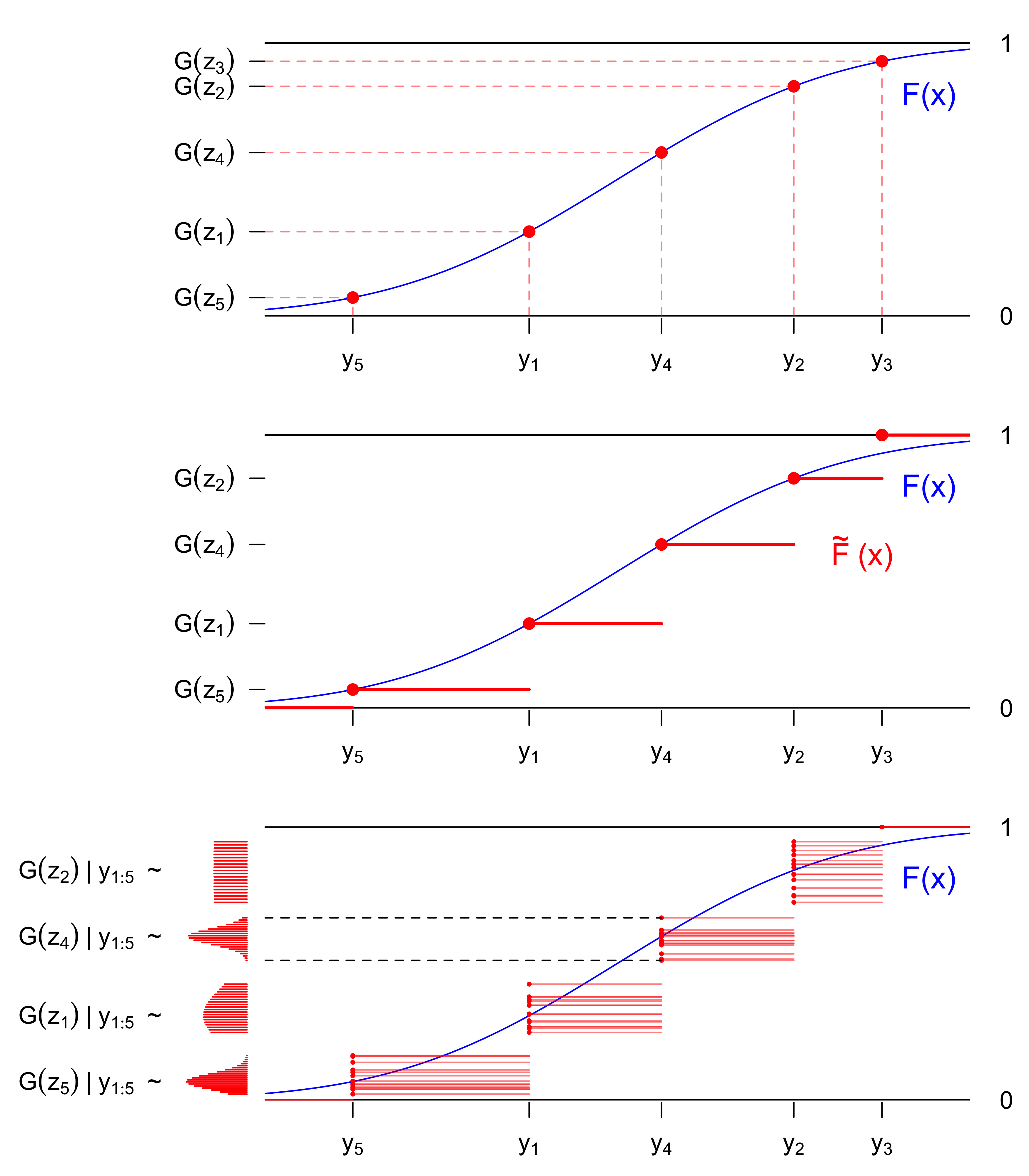}
    \caption{\small Given $z_t$, the unknown CDF $F$ (blue) must pass through the points $\{(y_t\com G(z_t))\}_{t=1}^T$ (top). The margin adjustment interpolates these points with a step function  $\tilde{F}$ (middle). But in practice the $G(z_t)$ are uncertain quantities with a posterior distribution inherited from $z_t$ (left margin, bottom), which then induces a posterior distribution for $\tilde{F}$ (bottom).} 
    \label{fig:ma}
\end{figure} 

\subsection{Putting it together}

Our strategy for approximating the joint posterior  $p(F_{1:n}\com\Btheta \given \By_{1:T})$ is outlined in Algorithm~\ref{alg:post}. 
\begin{algorithm}
{\bf Input:} Data $\By_{1:T}$
\begin{compactenum}  
    \item Sample $\{\latent^{(m)}_{1:T}\com\Btheta^{(m)}\}\sim p\{\latent_{1:T}\com\Btheta\given\latent_{1:T}\in\dcal(\By_{1:T})\}$;
    \item Compute $\tilde{F}_i^{(m)}$ using (\ref{eq:ma}) with $\latent^{(m)}_{1:T}$ for each $i=1,\ldots,n$.
\end{compactenum}
{\bf Output:} $\{\tilde{F}^{(m)}_{1:n}\com\Btheta^{(m)}\} \sim \tilde{p}(F_{1:n}\com\Btheta\given \By_{1:T})$.
\caption{Joint posterior sampling for the MTS copula  model (\ref{eq:ztrans}\com \ref{eq:stationary_link}).
 \label{alg:post}}
\end{algorithm}  
The first step follows Section~\ref{sec:erl} using the rank posterior while the second step follows Section~\ref{sec:ma} using the margin adjustment. The output $\{\tilde{F}^{(m)}_{1:n}\com\Btheta^{(m)}\}$ is drawn jointly from the pseudo-posterior in (\ref{eq:pseudo}). Although the algorithm does not explicitly store $\Bz_{1:T}^{(m)}$, we note that this quantity is generated by Algorithm~\ref{alg:post}  and subsequently used for probabilistic forecasting (Algorithm~\ref{alg:forecast}). 

While Section~\ref{sec:ma} leveraged the degeneracies of the \emph{conditional} posterior $p(F_{1:n}\given \latent_{1:T}\com\Btheta\com\By_{1:T})$ to construct the margin adjustment, Algorithm~\ref{alg:post} helps clarify how our scheme delivers posterior inference for the margins. Notably, the margin adjustment \eqref{eq:ma} is a function of the latent data $\latent_{1:T}$. Once draws of the $\latent_{1:T}^{(m)}$ are generated in the first step, the margin adjustment is updated accordingly in the second step, leading to draws $\tilde F_{1:n}^{(m)}$. 
This is visualized in Figure~\ref{fig:ma} (bottom). Put another way, the conditional posterior $p(F_{1:n}\given \latent_{1:T}\com\Btheta\com\By_{1:T})$ is degenerate at the interpolation points, but the marginal (over $\latent_{1:T}$) posterior for $F_{1:n}$ is not.

Finally, we provide probabilistic forecasting in Algorithm~\ref{alg:forecast}. The algorithm is written generally, but in practice we use the draws from the pseudo-posterior generated in  Algorithm~\ref{alg:post}. 
\begin{algorithm}
{\bf Input:} draws $\{\Bz_{1:T}^{(m)}, F_{1:n}^{(m)}, \Btheta^{(m)}\} \sim p(\Bz_{1:T}\com F_{1:n} \com \Btheta \given \By_{1:T})$
\begin{compactenum}  
\item Forward-sample $\{\Bz_{T+1:T+H}^{(m)} \} \sim p(\Bz_{T+1:T+H} \given \Bz_{1:T} = \Bz_{1:T}^{(m)}\com \Btheta = \Btheta^{(m)})$;
\item Compute $y_{t,i}^{(m)} = (F_i^{(m)})^{-1} \circ G_i(z_{t,i}^{(m)})$ for $t=T+1,\ldots,T+H$ and $i=1,\ldots,n$.
\end{compactenum}
{\bf Output:} $\{\By_{T+1:T+H}^{(m)}\} \sim p(\By_{T+1:T+H}\given \By_{1:T})$
\caption{Probabilistic forecasting for the MTS copula  model (\ref{eq:ztrans}\com \ref{eq:stationary_link}).
 \label{alg:forecast}}
\end{algorithm}  
Notably, Algorithm~\ref{alg:forecast} incurs minimal additional computational cost. The first step forward-simulates the latent MTS \eqref{eq:ztrans}. For the DGFC, this requires forward-simulating the dynamic factors in \eqref{eq:model_factor} and then generating the corresponding (multi-step-ahead) latent data in \eqref{eq:model_xz}. The second step simply applies the probability integral transform \eqref{eq:stationary_link}. 

Most important, Algorithm~\ref{alg:forecast} generates draws from the posterior predictive distribution \eqref{eq:pseudo-pred} and enables the full range of probabilistic (multi-step) forecasts. This is not typical for MTS copula models: existing methods either ignore the uncertainty from the margins $F_{1:n}$  \citep{ct2015joe, sv2016jbes} or impose parametric restrictions \citep{smith2015ijf}. Our approach does neither: we account for the uncertainty in the (unknown) margins, which are learned nonparametrically and automatically without any user tuning. Lastly, the combined strategy in Algorithms~\ref{alg:post}--\ref{alg:forecast} applies generally to provide joint posterior inference and probabilistic forecasting for MTS copula models (\ref{eq:ztrans}\com \ref{eq:stationary_link}). In practice, the modeler must decide on the latent MTS model \eqref{eq:ztrans}, which determines cross-sectional and serial dependencies along with computational scalability. We propose and evaluate the DGFC with latent MTS model (\ref{eq:model_xz}\com \ref{eq:model_factor}) as a  default option.


\subsection{Priors and posterior computation for the DGFC}\label{sec:priors}
We now discuss our default prior specifications for the DGFC. First, the parameters in \eqref{eq:model_xz} are the factor loadings $\BLambda=[\lambda_{i,l}]$ and the (diagonal) variances $\Bv=[v_1\,v_2\,\cdots\,v_n]^\tr$. For the factor loadings, we use the multiplicative gamma process prior \citep{bd2011bka}, which provides ordered shrinkage across the columns of $\BLambda$ to reduce sensitivity to the specified number of factors $k$. The variances are assigned $v_i^{-1}\sim\text{iid}\,\text{Gamma}(\alpha_0\com\beta_0)$. Next, we use a conditionally  conjugate prior for the VAR parameters in \eqref{eq:model_factor},  $\BG^\tr\com\BSigma\sim\text{MNIW}_{k,k}(d_0\com\BPsi_0\com\overline{\BG}_0^\tr\com\BO^{-1}_0)$,  truncated so that $\BG$ always has eigenvalues within the unit circle for stationarity. In each of the simulations and applications that follow, we use the same default prior hyperparameters: $k=\lceil0.7 n\rceil$, $d_0=k+1$, $\BPsi_0=\BI_k$, $\overline{\BG}_0=\Bzero$, $\BO_0=\BI_k$, $\alpha_0=1$, $\beta_0=0.3$. Further details are in the supplement.

These priors are all conditionally conjugate within (\ref{eq:model_xz}\com \ref{eq:model_factor}), which leads to straightforward and efficient sampling steps for the main blocks outlined in Section~\ref{sec:erl}. For instance, $(\BG\com\BSigma)$ can be drawn using standard conjugacy results for VARs \citep{karlsson2013chapter}, while $\Beta_{1:T}$ appears within a linear, (conditionally) Gaussian state space system (\ref{eq:model_xz}\com \ref{eq:model_factor}) and can be drawn using the Kalman simulation smoother. Finally, the latent data $\Bz_{1:T}$ are  conditionally independent given $\{\Beta_{1:T}, \BLambda, \Bv\}$, which removes all matrix inversions from their truncated sampling steps (Section~\ref{sec:erl}) and dramatically improves computational scalability.  Details on the Gibbs sampling algorithm are in the supplementary material (Algorithm~\ref{alg:gibbs_sampler}).

\section{Theoretical results}\label{sec:theory}




We study the asymptotic behavior of the pseudo-posterior $\tilde{p}(F_{1:n}\com\Btheta\given\By_{1:T})$ in (\ref{eq:pseudo}). 
To begin, we state a technical result demonstrating that the margin adjustment  (\ref{eq:ma}) approximates the true distribution in the context of a model like (\ref{eq:ztrans}\com \ref{eq:stationary_link}). All proofs are in the supplementary material. 
\begin{theorem}\label{thm:ma}
    Let $Z_t$ be a strictly stationary ergodic process with continuous marginal distribution $G(x)=P(Z_1\leq x)$, and let $F$ be any continuous distribution function. If we have $Y_t=F^{-1}\circ G(Z_t)$ and define $\tilde{F}$ as in (\ref{eq:ma}), then $\tilde{F}(x)\cas F(x)$ for all $x\in\RR$ as $T\to\infty$.
\end{theorem}



Next, we 
show that the rank posterior $p\{\Btheta\given\latent_{1:T}\in\dcal(\By_{1:T})\}$ provides consistent inferences for the MTS copula parameters. We do so under a Gaussian VAR for the latent MTS model \eqref{eq:ztrans}:
\begin{equation}
     \latent_{t}=\tilde{\BG}\latent_{t-1}+\Bepsilon_t,\quad\Bepsilon_t\iid\text{N}_n(\Bzero\com\tilde{\BSigma})\label{eq:simple_latent}
\end{equation}
Posterior consistency requires parameter identification, and the copula parameters $\Btheta=\{\tilde{\BG}\com\tilde{\BSigma}\}$ are only identified if $\tilde{\BG}$ has eigenvalues within the unit circle, and if the stationary covariance matrix given by $\text{vec}(\BC_0)=(\BI_{n^2}-\tilde{\BG}\otimes\tilde{\BG})^{-1}\text{vec}(\tilde{\BSigma})$ is a correlation matrix. This also implies that the margin (or stationary distribution) of \eqref{eq:simple_latent} is $G_i = \Phi$. 
We discuss this further in Appendix~\ref{app:id}, but for the purposes of the theorems, we assume that our prior for $\Btheta$ is supported only on the subset $\Theta\subseteq \RR^{n\times n}\times\text{SPD}_n$ where these conditions hold. In addition, we require the following:
\begin{compactenum}[({A}.1)]
\item $\By_{1:\infty}$ is distributed according to $H^{\infty}_{\Btheta_0,F_{1:n}}$, the joint distribution implied by (\ref{eq:stationary_link}\com\ref{eq:simple_latent});
\item $F_i$ are all continuous;
\item $\Pi$ is a prior distribution on $\Theta$ and $\pi(\Btheta)$ is its density with respect to some measure $\nu$; 
\item $\pi(\Btheta)>0$ a.e.\ with respect to $\nu$.
\end{compactenum}

\begin{theorem}\label{thm:doob1}
Under (A.1)--(A.4),  
for $\Btheta_0$ a.e.\ $\nu$ and for any neighborhood $A\ni\Btheta_0$, we have \\
$
\lim_{T\to\infty} \Pi\{\Btheta\in A\given \latent_{1:T}\in\dcal(\By_{1:T})\}=1\text{ a.s. }H^{\infty}_{\Btheta_0,F_{1:n}}.$
\end{theorem}


This result implies that the ignored ratio term in (\ref{eq:rank}) becomes negligible as $T$ grows. Lastly, we show that our pseudo-posterior for the margins (\ref{eq:ma}) provides consistent inferences. We use $\tilde{\Pi}$ to denote the marginal margin adjustment posterior distribution implied by (\ref{eq:pseudo}):

\begin{theorem}\label{thm:doob2}
Under (A.1)--(A.4), define $\tilde{F}_i$ as in (\ref{eq:ma}). For any $x\in\RR$ and any neighborhood $A\ni F_i(x)$, we have
    $\lim_{T\to\infty} \tilde{\Pi}\{\tilde{F}_i(x)\in A\given \By_{1:T}\}=1\text{ a.s. }H^{\infty}_{\Btheta_0,F_{1:n}}.$
\end{theorem}



Taken together, these results provide asymptotic validation for our posterior approximation strategy via the pseudo-posterior $\tilde{p}(F_{1:n}\com\Btheta\given\By_{1:T})$ in (\ref{eq:pseudo}). 
We confirm these theoretical results using simulations in the supplementary material (Appendix~\ref{app:sim}). 

The main limitations here are 1) these results assume model correctness, both of the MTS copula model  (\ref{eq:ztrans}\com \ref{eq:stationary_link}) and the specific latent dynamics of \eqref{eq:simple_latent}, and 2) they use continuous margins. Our empirical evaluations with  simulated (Section~\ref{sec:sims}) and real (Section~\ref{sec:applications}) data remove these limitations and consider misspecified models with continuous and discrete MTS data.

\section{Competing methods}\label{sec:other}


We advocate the proposed approach as a versatile, general-purpose utility for probabilistic  forecasting of MTS data. To establish this, we compare the DGFC against  benchmark competitors that are used for probabilistic forecasting of MTS across a broad range of scenarios with minimal tuning. Specifically, we focus on Bayesian VARs (\textbf{BVAR}) and dynamic linear models ({\bf DLM}).  Bayesian VARs apply \eqref{eq:simple_latent} to $\By$ instead of $\Bz$ and thus have a (conditionally) Gaussian marginal distribution for each variable. We incorporate the estimation procedure and shrinkage priors from \cite{glp2015restat}, which jointly models the VAR coefficients and error covariance matrix using a  a conjugate matrix-normal-inverse-Wishart prior with ``Minnesota''-type shrinkage toward a random walk for each variable. We adopt the default prior settings from the   \texttt{BVAR} package \citep{kv2021jss} in \texttt{R}. Next, we consider DLMs (or state space models) of the form 
    \begin{align}
        \By_t&=\Bmu_t+\Bepsilon_t, &&\Bepsilon_t\iid\text{N}_n(\Bzero\com \BW)\label{eq:measurement}\\
        \Bmu_t&=\Bmu_{t-1}+\Bxi_t,&&\Bxi_t\iid\text{N}_n(\Bzero\com\BV) \label{eq:evolution}
    \end{align}
    and implemented in the \texttt{dlm} \citep{petris2010jss} and \texttt{KFAS} \citep{helske2017jss} packages in \texttt{R}. The latter also supports dynamic \textit{generalized} linear models (DGLMs) for modeling non-Gaussian data: DGLMs model count data by replacing   (\ref{eq:measurement}) with  $y_{t,i}\given \mu_{t,i}\sim \text{Poisson}(\exp(\mu_{t,i}))$ or a negative-binomial distribution, with serial and cross-sectional dependencies modeled via \eqref{eq:evolution}. 
     By default, the static parameters $\BW$ (when included) and $\BV$ are estimated via maximum likelihood. 

\section{Simulation results}\label{sec:sims}




We evaluate estimation, inference, and probabilistic forecasting using simulated data. To reflect practical usage, and to complement our theoretical analysis (Section~\ref{sec:theory}), we focus on misspecified models that wholly or partly deviate from the DGFC. We consider three data-generating processes (DGPs), each of which yields different cross-sectional and serial dependencies along with different marginal  stationary  distributions. For each of these DGPs, the margins are known exactly, which allows us to evaluate estimation and inference for the proposed margin adjustment. However, our approach does \emph{not} assume knowledge of any marginal stationary distributions. 

The DGPs are the following:
\begin{compactitem}
\item \textbf{VARMA}:  $
        \By_t=\Bb_0+\sum\limits_{\ell=1}^p\BB_\ell\By_{t-\ell}+\sum\limits_{j=1}^q\BC_j\Bepsilon_{t-j}+\Bepsilon_t$ with $ \Bepsilon_t\iid\text{N}_n(\Bzero\com\BSigma).
    $ Subject to conditions on the parameter matrices, this is a stationary Gaussian process and the marginal moments can be computed exactly 
    \citep{lutkepohl2006book}. 
    We set $(n\com p\com q)=(2\com3\com6)$ and draw random parameters $\BSigma\sim\text{IW}_n(n+1\com\BI_n)$, $\BC_l\sim\text{iid}\,\text{N}(0\com1)$, $\Bb_0\sim\text{N}_n(\Bzero\com\BI_n)$, and $\BB_l\sim\text{iid}\,\text{N}(0\com0.1)$, redrawing the $\{\BB_\ell\}_{\ell=1}^p$ until they imply stationarity.
\item \textbf{VARCH}: $ \By_t\given\By_{t-1}\sim t_n\{\nu\com\Bzero\com (\BA+\By_{t-1}\By_{t-1}^\tr)/\nu\}$ with $\nu>n+2$  is a multivariate generalization of the autoregressive conditional heteroskedasticity model based on the multivariate-$t$ distribution \citep{pw2005jasa}, 
     with  stationary distribution $\By_t\sim t_n\{\nu-1\com\Bzero\com\BA/(\nu-1)\}$. This DGP is challenging: the dynamics are nonlinear and the stationary distributions have heavy tails.    We set $n=2$ and draw $(\nu-n-2)\sim\text{Gamma}(1, 1)$ and $\BA\sim\text{IW}_n(n+1\com\BI_n)$.
\item \textbf{VARMA Copula}: we replace the DFM in (\ref{eq:model_xz}\com \ref{eq:model_factor}) with  $\unscaled_t \sim\text{VARMA}_n(p\com q)$ for $(n\com p\com q)=(2\com1\com1)$ and specify a heterogeneous mix of marginal distributions, $F_1=\text{Poisson}(\lambda = 5)$ and $F_2=\text{Skew-}t(\text{df}=3\com\mu=0\com\sigma=1\com\alpha=2)$ from \cite{ac2003jrssb}. The VARMA parameters are simulated as in the {\bf VARMA} case above. 
\end{compactitem}



    These DGPs present significant challenges: in each case, the DGFC has misspecified dynamics and no knowledge about the marginal stationary distributions. In that context, we  seek to answer two main questions:
\begin{compactenum}
    \item Does the margin adjustment posterior provide accurate estimates with precise and well-calibrated uncertainty quantification for unknown marginal stationary distributions?
    \item Does our (pseudo-)posterior predictive distribution provide accurate probabilistic forecasts compared to popular MTS alternatives?
\end{compactenum}

\subsection{Recovering the marginal stationary distributions}

Figure~\ref{fig:simulation_recovery} plots the true, marginal stationary distribution for a representative variable from each of the three DGPs, along with credible bands visualizing the margin adjustment posterior across increasing sample sizes $(T=25\com50\com100\com200)$. The posterior approximation uses Algorithm~\ref{alg:post} with 1,000 posterior draws retained from 10,000 (burning the first 5,000 and thinning every 5th). We use the same default priors in all cases. 

For each DGP, the margin adjustment posterior concentrates around the true stationary distribution as $T$ grows. This is similar to the behavior guaranteed by Theorem~\ref{thm:doob2}, but here our MTS copula model is misspecified. Notably, convergence is typically quite fast: $T=50$ or $T=100$ is enough to achieve strong concentration around the ground truth. Taken together, these results demonstrate the ability of the margin adjustment posterior to quickly and automatically learn the marginal stationary distributions of MTS exhibiting a range of distributional features (discrete support, heavy tail, skew, etc.) and across a wide variety of underlying dynamics. 

\begin{figure}[h]
    \centering
    \setlength{\tabcolsep}{0pt}
    \begin{tabular}{cccc} \small
        VARMA & VARCH & VARMA Copula & VARMA Copula  \\
        (Gaussian variable) & ($t$ variable) & (Poisson variable) & (Skew-$t$ variable) \\
       \includegraphics[width=0.18\textwidth]{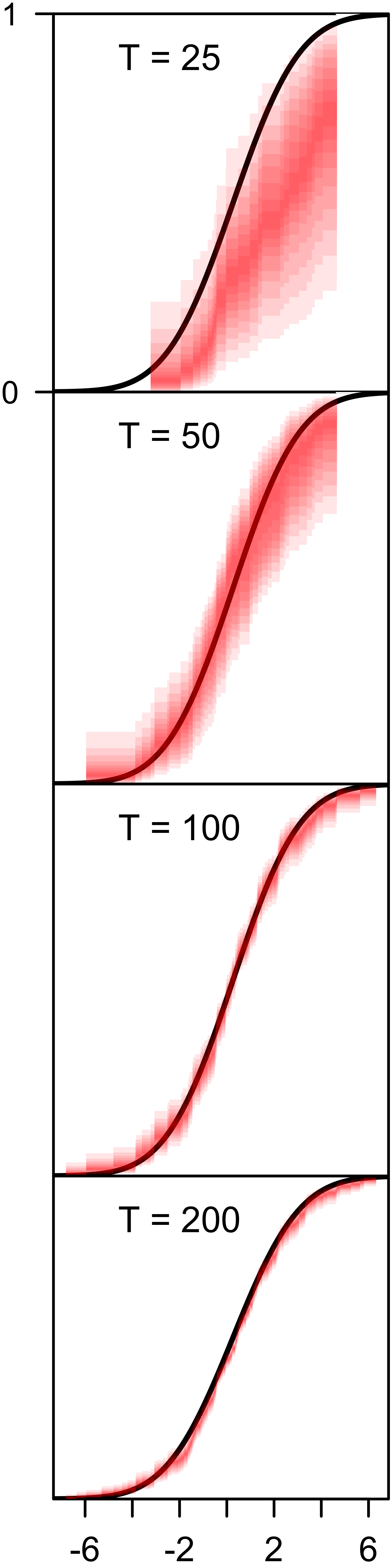}  
       & 
       \includegraphics[width=0.18\textwidth]{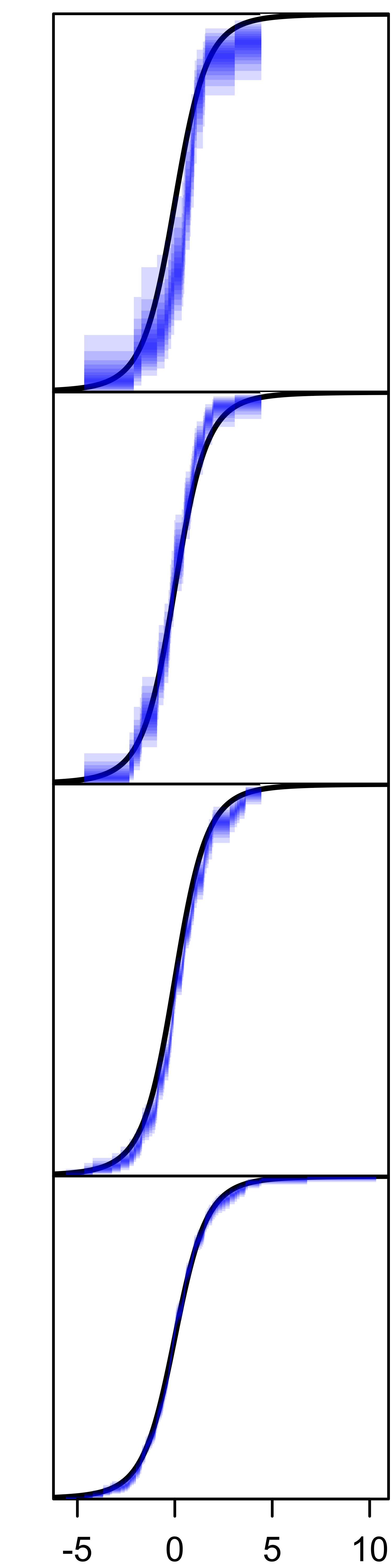}
       &
       \includegraphics[width=0.18\textwidth]{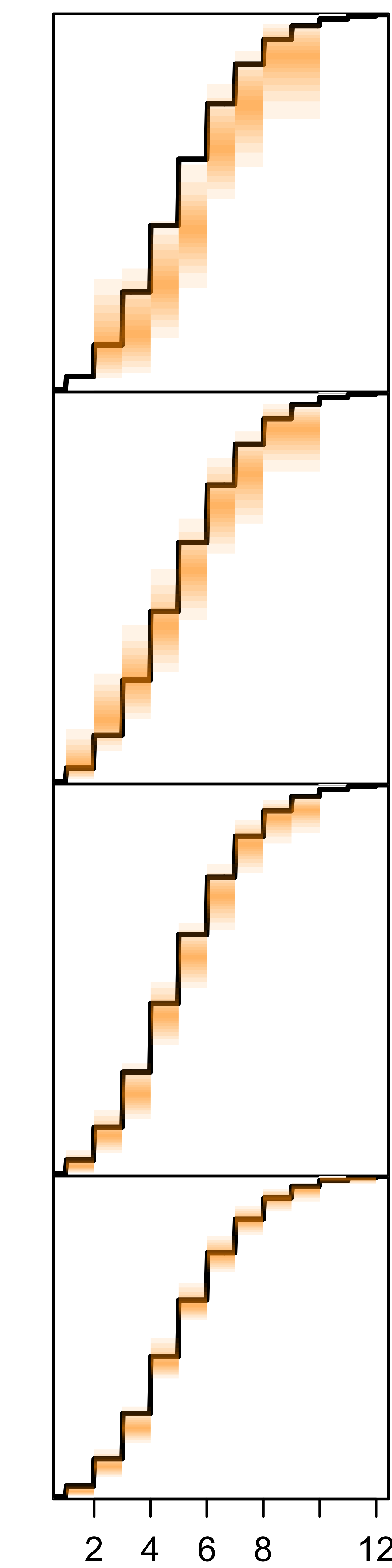}
       &
       \includegraphics[width=0.18\textwidth]{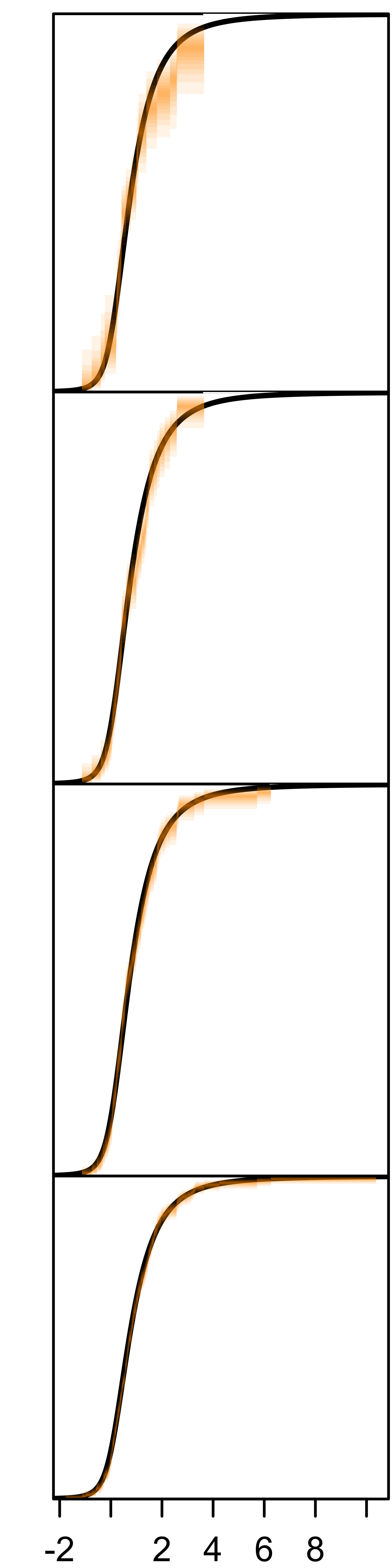}
    \end{tabular}
    \caption{\small Pointwise posterior credible bands ($\alpha=0.1\com0.2\com...\com0.9$) based on the sample quantiles of  $\{\tilde{F}^{(m)}\}_{m=1}^{1000}$ drawn from Algorithm~\ref{alg:post}. Isolating a representative variable for each DGP, we see the margin adjustment posterior concentrating around each true stationary distribution as $T$ grows.} 
    \label{fig:simulation_recovery}
\end{figure}






\subsection{Probabilistic forecast evaluation}\label{sec:forecastsim}
We next compare the probabilistic forecasts from the proposed DGFC against the competitors from Section~\ref{sec:other}  using simulated data from our three DGPs. We do this with a recursive, out-of-sample forecasting exercise. We first simulate a time series of length $T$ from a DGP. Then, we recursively refit each of our models on a expanding window of this data. So for each time point $t = t_0,\ldots,T$, where $t_0>h$ is some starting period for the exercise, we retrain every model (the DGFC, BVAR, and DLM) on the initial subset of data $\By_{1:t-h}$ and then generate point, interval, and density forecasts of the next $h=1\com 2\com ...\com 10$ periods. We compare the forecasts $\hat{\By}_{t+h}$ to the subsequent realizations $\By_{t+h}$, and average the forecasting performance over time for each model. 

We use three metrics to summarize forecasting accuracy. Point forecasting accuracy is measured using mean squared error for continuous variables and mean absolute error for count variables. Interval forecasts, computed as 95\% highest posterior predictive density regions approximated from the posterior draws, are evaluated by average size and empirical coverage. Finally, we evaluate the entire density forecast using the mean continuous ranked probability score \citep{gr2007jasa}. The metrics are averaged over a period from $t_0=10$ to $T=300$.

Figure~\ref{fig:simulation_forecasting} displays the forecasting results for each model and each horizon $h$ across the three DGPs.  Despite the model misspecifications, the DGFC consistently delivers the best density forecasts, highly competitive point forecasts, and the narrowest interval forecasts (although slightly below the nominal coverage) across all $H=10$ horizons.  The same model configurations and default priors are used in all cases, so we again observe the versatility and adaptability of our framework across a range of MTS with differing dynamics and distributional features.

\begin{figure}[h]
    \centering
    \setlength{\tabcolsep}{0pt} \small 
        \begin{tabular}{cccc}
        VARMA & VARCH & VARMA Copula & VARMA Copula  \\
        (Gaussian variable) & ($t$ variable) & (Poisson variable) & (Skew-$t$ variable) \\
    \includegraphics[width=0.24\textwidth]{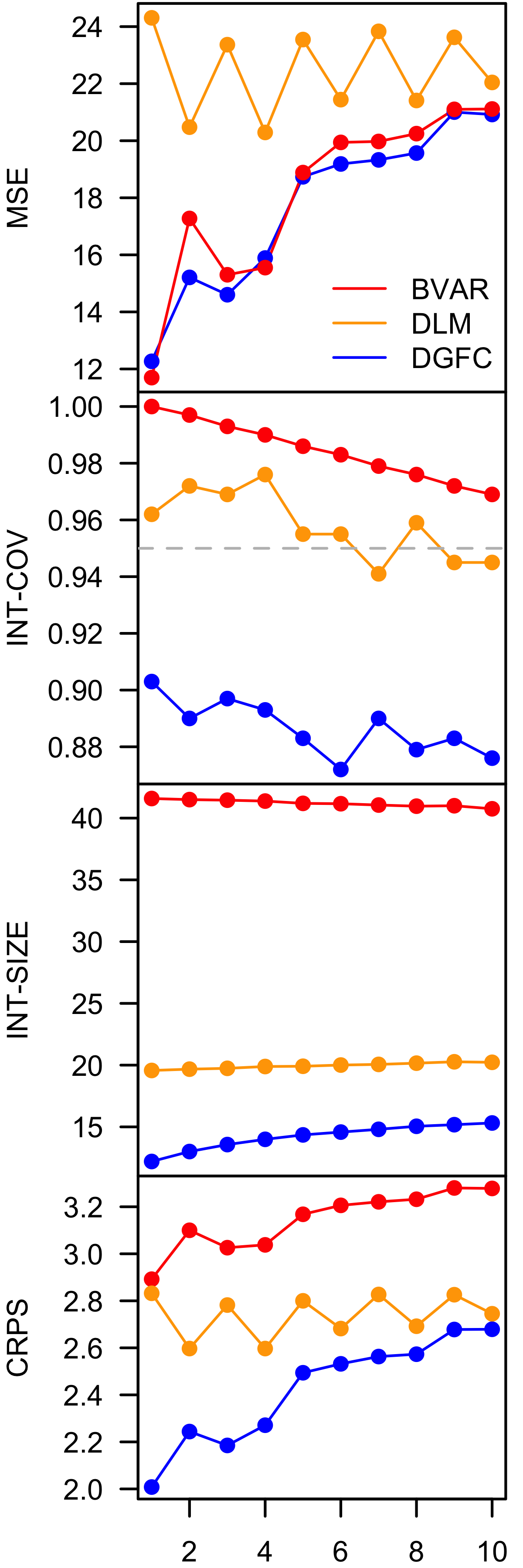}
    &
    \includegraphics[width=0.24\textwidth]{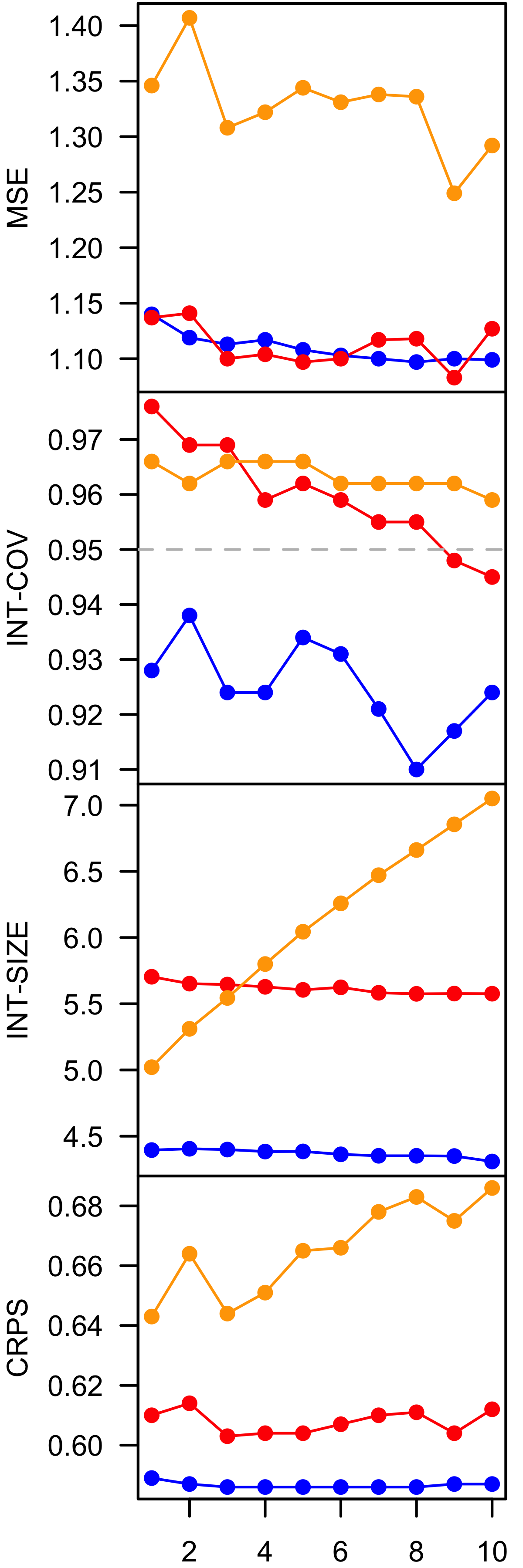}
    & 
    \includegraphics[width=0.24\textwidth]{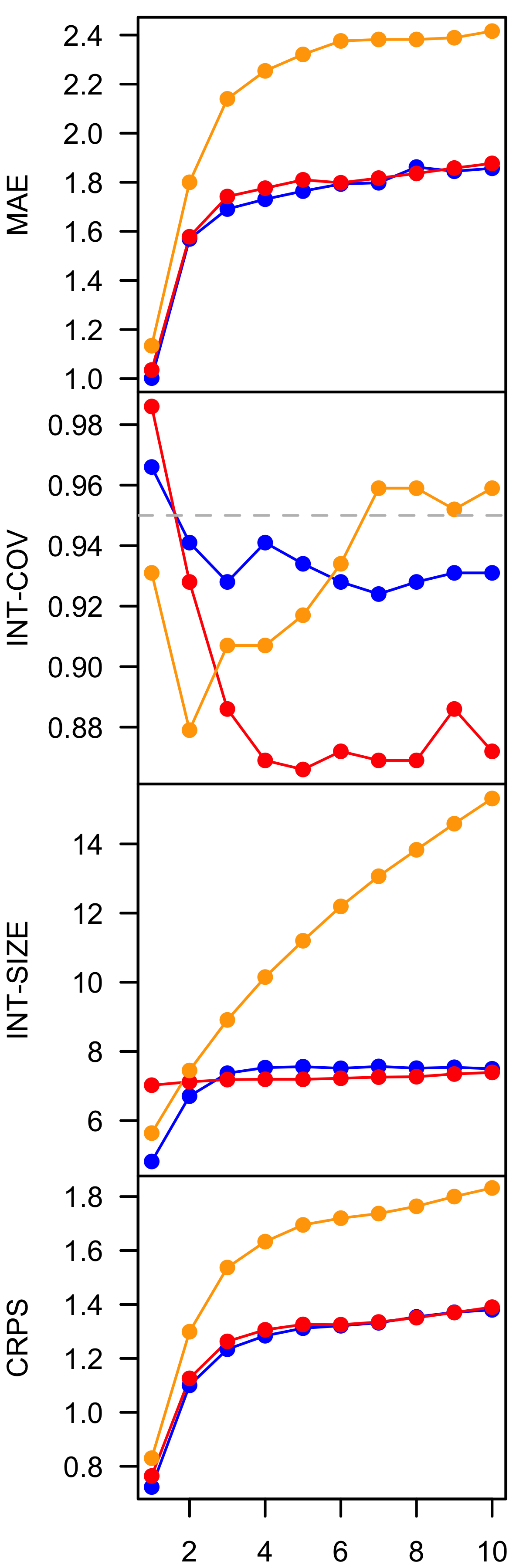}
    & \includegraphics[width=0.24\textwidth]{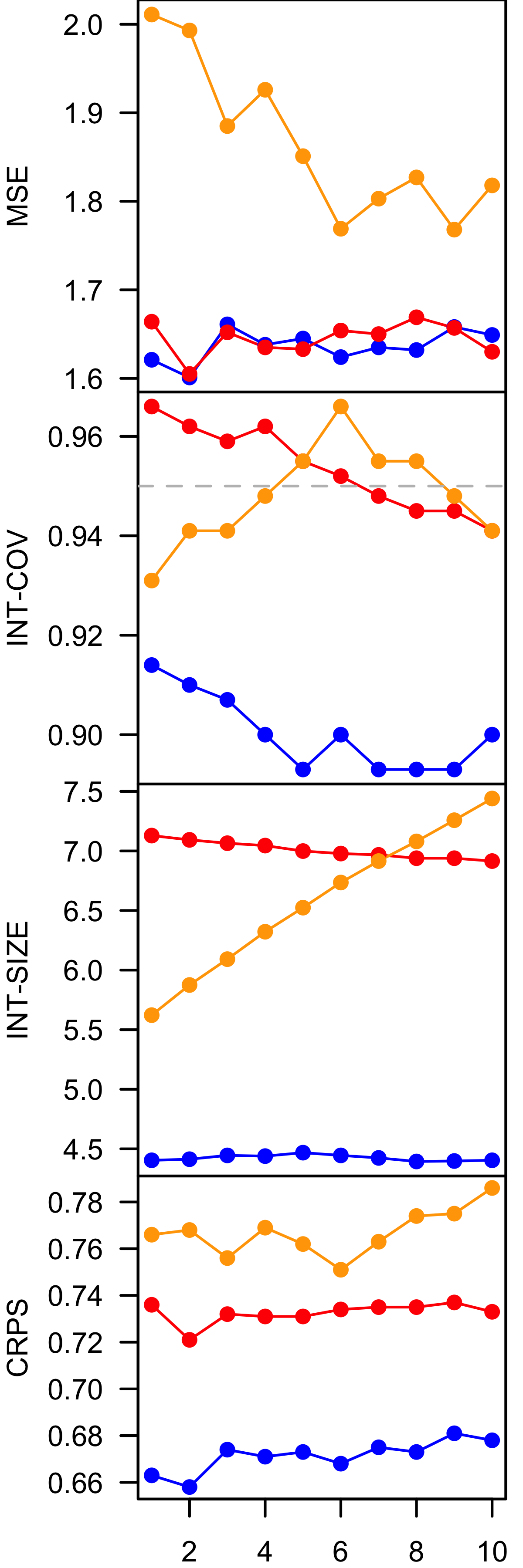}
    \\
     \multicolumn{4}{c}{horizon} 
    \end{tabular}
    \caption{\small Forecast metrics for horizons $h=1\com 2\com ...\com 10$, averaged over a synthetic dataset of length $T=300$ from each DGP. MSE and MAE are the mean squared and mean absolute error of the posterior predictive mean and median, respectively. Forecast intervals are 95\% highest posterior predictive density sets for continuous variables, and equal-tailed quantile intervals for discrete variables.}
    %
    %
    \label{fig:simulation_forecasting}
\end{figure}

\section{Real-data applications}\label{sec:applications}


\subsection{Crime counts in New South Wales, Australia}\label{sec:crime}

The New South Wales Bureau of Crime Statistics and Research releases data counting the number of monthly recorded cases for a variety of crimes\footnote{New South Wales Bureau of Crime Statistics and Research (\texttt{bcsr@justice.nsw.gov.au}): \url{https://bocsar.nsw.gov.au/statistics-dashboards/open-datasets/criminal-offences-data.html}}. These data constitute a count-valued MTS  where each  $y_{t,i}$ is the number of times crime $i$ was recorded in period $t$ (January 1995 through December 2023). Figure~\ref{fig:crime} displays the data for four of these variables (the rest are in the supplementary material), which can differ significantly in their distributional features. As such, this example provides an interesting test case for our approach. In our exercise, we focus on jointly modeling the $n=10$ variables listed in Table~\ref{tab:crime}. 

Figure~\ref{fig:crime_params} displays summaries of the parameter inferences from our model. On the left we plot the margin adjustment posterior for a selection of variables that exhibit  heterogeneous features (small counts, heavy tails, asymmetry). Using the empirical cumulative distribution function (ECDF) as a reference point, the margin adjustment appears adequate while also providing reasonable uncertainty quantification.  On the right, we provide heat map posterior summaries of  the stationary correlation $\text{corr}(\Bz_t)$ and first-lag cross-correlation $\text{corr}(\Bz_{t+1}\com\Bz_t)$ on the latent scale, which are proxies for the cross-sectional and serial dependence in $\By_t$ (defined in the supplementary material, Appendix~\ref{sec:dgfc_dependence}). The model identifies  strong positive (cross-sectional and serial) correlations among murder, attempted murder, manslaughter, and escaping custody. Blackmail is negatively (cross-sectionally and serially) correlated with murder, attempted murder, and escaping custody. In general, cross-sectional correlations persist serially, while autocorrelations are strong for each individual series. 




\begin{figure}[h]
    \centering
    \includegraphics[width = 0.27\textwidth]{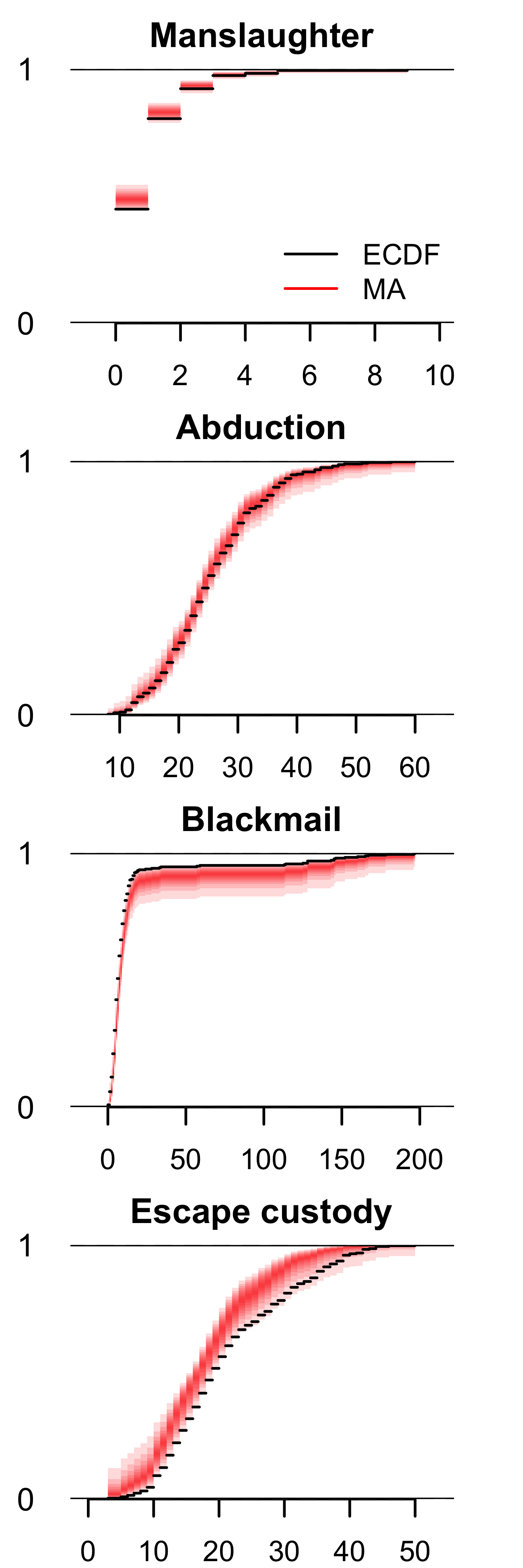}
    \includegraphics[width = 0.57\textwidth]{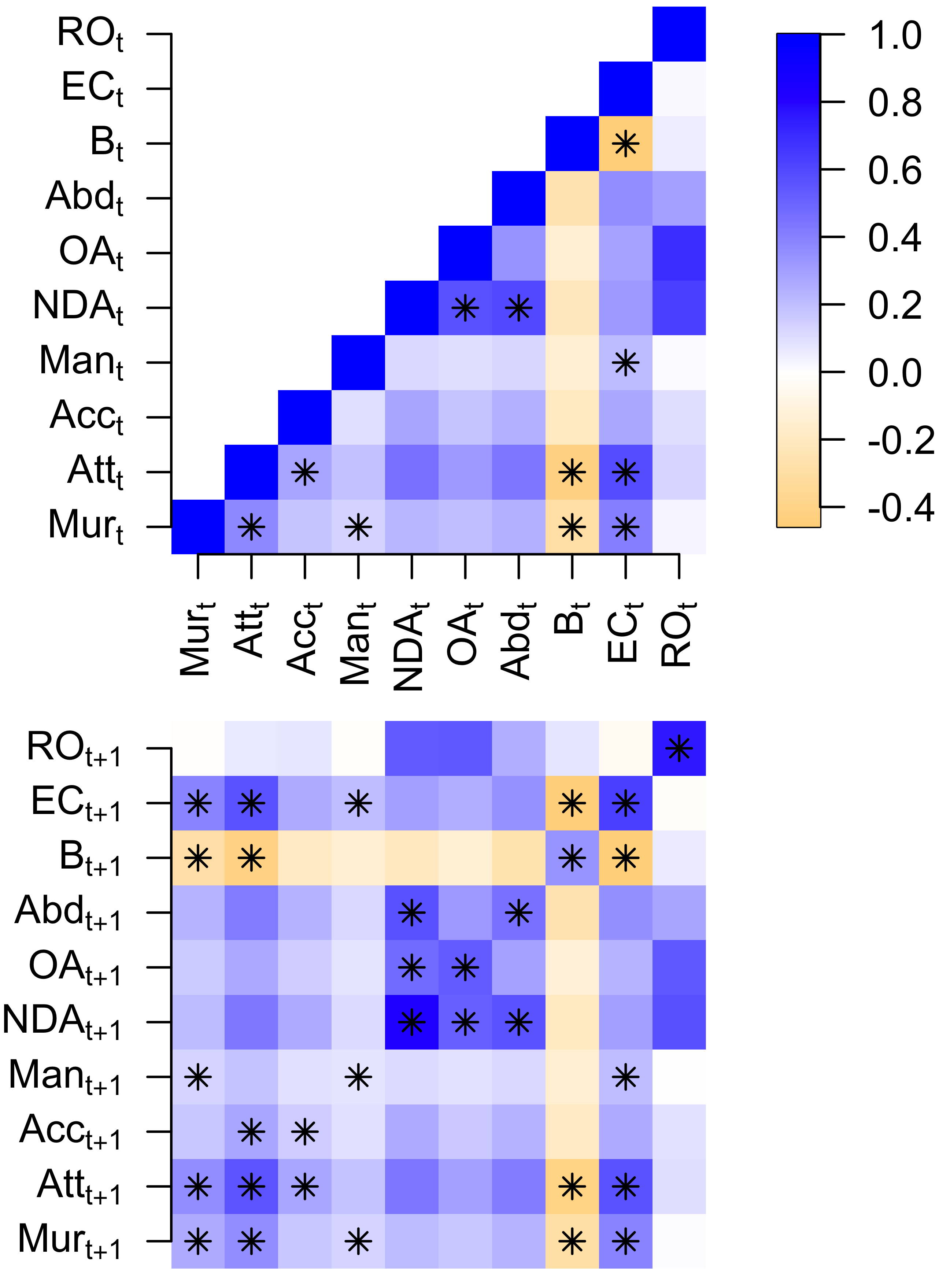}
    \caption{\small \textbf{Left}: the margin adjustment posterior (pointwise quantile bands of MCMC draws) and the raw ECDF for select crime variables. \textbf{Right}: posterior mean estimate of $\text{corr}(\Bz_t)$ (top) and $\text{corr}(\Bz_{t+1}\com\Bz_t)$ (bottom). A ``*'' indicates that the 95\% (with Bonferroni correction) highest posterior density interval excludes zero. \textbf{Abbreviations}: Murder (Mur), Attempted murder (Att), Accessory (Acc), Manslaughter (Man), Non-domestic assault (NDA), Assaulting an officer (OA), Abduction (Abd), Blackmail (B), Escaping Custody (EC), Resisting an officer (RO).}
    \label{fig:crime_params}
\end{figure}

Table~\ref{tab:crime} presents  one-month-ahead forecast metrics for each variable, averaged over the period January 2000 ($t_0$) to December 2023 ($T$). We compare our approach to  BVAR($p=1$) and two count-valued multivariate state space models: a Poisson DLGM (P-DGLM) and a negative binomial DGLM (NB-DGLM), both using the dynamics from \eqref{eq:evolution}.  For all variables, the DGFC delivers the best density forecasts and the smallest 95\% interval forecasts with near nominal coverage. For nine of ten variables, the DGFC provides superior point forecasts, often by a large amount. 

\setlength{\tabcolsep}{6pt}
\begin{table}[h] \small
    \centering
    \begin{tabular}{l  rrrr  r rrrr}
               & \textbf{MAE}  &  \textbf{COV}  &  \textbf{SIZE}  &  \textbf{CRPS}  &  & \textbf{MAE}  &  \textbf{COV}  &  \textbf{SIZE}  &  \textbf{CRPS} \\
    \cline{2-5}
    \cline{7-10}
     & \multicolumn{4}{|c|}{Murder}  &  & \multicolumn{4}{|c|}{Assaulting an officer}\\\cline{2-5}\cline{7-10}
{DGFC} & \textbf{2.16} & 0.96 & \textbf{11.23} & \textbf{1.53} & & \textbf{24.34} & 0.93 & \textbf{115.29} & \textbf{17.40} \\
{BVAR} & 3.20 & 1.00 & 437.63 & 18.37 & & 109.93 & 1.00 & 14683.04 & 849.64 \\
{P-DGLM} & 15.32 & 0.85 & 73.45 & 13.30 & & 368.79 & 0.78 & 1547.76 & 292.29 \\
{NB-DGLM} & 13.88 & 0.92 & 152.19 & 11.29 & & 312.50 & 0.93 & 3539.22 & 247.04 \\\cline{2-5}\cline{7-10}
     & \multicolumn{4}{|c|}{Attempted murder}  &  & \multicolumn{4}{|c|}{Abduction}\\\cline{2-5}\cline{7-10}
DGFC & \textbf{1.92} & 0.98 & \textbf{10.49} & \textbf{1.33} && \textbf{5.40} & 0.96 & \textbf{29.54} & \textbf{3.80} \\
BVAR & 5.12 & 1.00 & 774.11 & 38.94 & &14.10 & 1.00 & 1947.94 & 102.81 \\
P-DGLM & 4.55 & 0.89 & 18.82 & 3.50 && 23.85 & 0.83 & 89.04 & 18.52 \\
NB-DGLM & 4.20 & 0.93 & 43.09 & 3.40 & &27.16 & 0.98 & 301.17 & 21.70 \\\cline{2-5}\cline{7-10}
     & \multicolumn{4}{|c|}{Accessory to murder}  &&  \multicolumn{4}{|c|}{Blackmail}\\\cline{2-5}\cline{7-10}
DGFC &\textbf{ 0.39} & 0.99 & \textbf{2.66} & \textbf{0.30} && 11.51 & 0.95 & \textbf{30.03} & \textbf{9.13} \\
BVAR & 0.70 & 1.00 & 96.30 & 3.55 && \textbf{5.25} & 1.00 & 391.85 & 16.85 \\
P-DGLM & 0.96 & 0.96 & 5.44 & 0.84 & &37.60 & 0.83 & 97.04 & 24.15 \\
NB-DGLM & 0.72 & 0.98 & 7.43 & 0.58 && 19.92 & 0.92 & 213.46 & 16.24 \\\cline{2-5}\cline{7-10}
     & \multicolumn{4}{|c|}{Manslaughter}  & & \multicolumn{4}{|c|}{Escaping custody}\\\cline{2-5}\cline{7-10}
DGFC & \textbf{0.74} & 0.99 & \textbf{3.48} & \textbf{0.50} && \textbf{4.73} & 0.94 & \textbf{21.43} & \textbf{3.29} \\
BVAR & 0.98 & 1.00 & 125.96 & 4.68 && 8.45 & 1.00 & 1194.11 & 57.35 \\
P-DGLM & 1.27 & 0.93 & 5.24 & 0.95 && 27.66 & 0.82 & 125.57 & 21.40 \\
NB-DGLM & 1.05 & 0.97 & 8.41 & 0.78 & &39.01 & 0.94 & 395.19 & 30.50 \\\cline{2-5}\cline{7-10}
     & \multicolumn{4}{|c|}{Non-domestic assault}  &  &\multicolumn{4}{|c|}{Resisting an officer}\\\cline{2-5}\cline{7-10}
DGFC & \textbf{214.8} & 0.93 & \textbf{979.9} & \textbf{150.8} && \textbf{58.2} & 0.95 & \textbf{295.1} & \textbf{41.7} \\
BVAR & 1799.0 & 1.00 & 232293.5 & 13731.8 && 261.5 & 1.00 & 36659.8 & 2138.9 \\
P-DGLM & 2612.8 & 0.75 & 8940.7 & 2046.7 && 523.9 & 0.75 & 1587.9 & 394.9 \\
NB-DGLM & 3290.3 & 0.96 & 34814.6 & 2610.1 && 572.4 & 0.95 & 5793.4 & 446.9 \\
    \end{tabular}
    \caption{\small One-month-ahead forecast metrics for the crime count data in Figure~\ref{fig:crime}, averaged over the period January 2000 to December 2023. \textbf{MAE} is the mean absolute error of the posterior predictive median (smaller is better). \textbf{COV} and \textbf{SIZE} are the empirical coverage and average length of $95\%$ credible intervals based on posterior predictive quantiles (prefer small intervals with high coverage). \textbf{CRPS} is continuous ranked probability score of the entire forecast distribution (smaller is better). The best (smallest) values of MAE, SIZE, and CRPS are bolded for each variable. The proposed method (DGFC) consistently provides the best density forecasts, the smallest interval forecasts with good coverage, and superior point forecasts. 
}
    \label{tab:crime}
\end{table}

To evaluate the impact of the cross-section size $n$, we repeat this exercise  using only the first three variables (murder, attempted murder, and accessory) that were considered in \cite{ls2019jcgs} (see the supplemental material). The DGFC still performs as well or better than the competitors, but the gap is smaller. Most notably, the forecasting performance of the competing methods deteriorates substantially from $n=3$ to $n=10$, while the DGFC forecasting performance remains excellent in both cases.  These results highlight the advantages of the DGFC for both \emph{probabilistic} and \emph{multivariate} time series forecasting, especially for moderate cross-sections. 

\subsection{US macroeconomic aggregates}


Forecasting macroeconomic time series (e.g., Figure~\ref{fig:macro}) is an important part of the decision making process within central banks and financial institutions, and there is a massive econometric literature devoted to it. Recent work has emphasized the importance of capturing   asymmetry and heavy tails when forecasting variables like GDP growth (see \citealt{sv2016jbes}, \citealt{ccm2024jmcb} and references therein), and so this example provides a useful test case for our method. As such, we consider the $n=14$ macroeconomic series in \cite{ccm2016jbes}, and mimic their real-time forecasting exercise to compare our approach with the BVAR and DLM competitors (see Section~\ref{sec:other}). Macroeconomic time series data exhibit the unique feature of being revised several times after their initial release, sometimes quite substantially, so when embarking upon a forecasting exercise like the one described in Section~\ref{sec:forecastsim}, it is important to take into account the version or \textit{vintage} of the data that would have been available in the period that a given forecast was being made. This best measures how a model would have performed historically if it had been deployed in real-time. \cite{ccm2016jbes} describe the implementation details that we use here (see the supplemental materials for more information about the dataset).

Figure~\ref{fig:macro_params} displays summaries of the parameter inferences in our model. Again, the DGFC  captures the non-Gaussian marginals and provides reasonable uncertainty quantification.  The estimated latent correlations identify  high positive associations the real economic activity variables (output, productivity, etc). We also observe strong association contemporaneously and intertemporally amongst the funds rate and the two measures of inflation. These stylized results match various basic economic intuitions about the dependence structure of these variables. 

\begin{figure}[h]
    \centering
    \includegraphics[width = 0.27\textwidth]{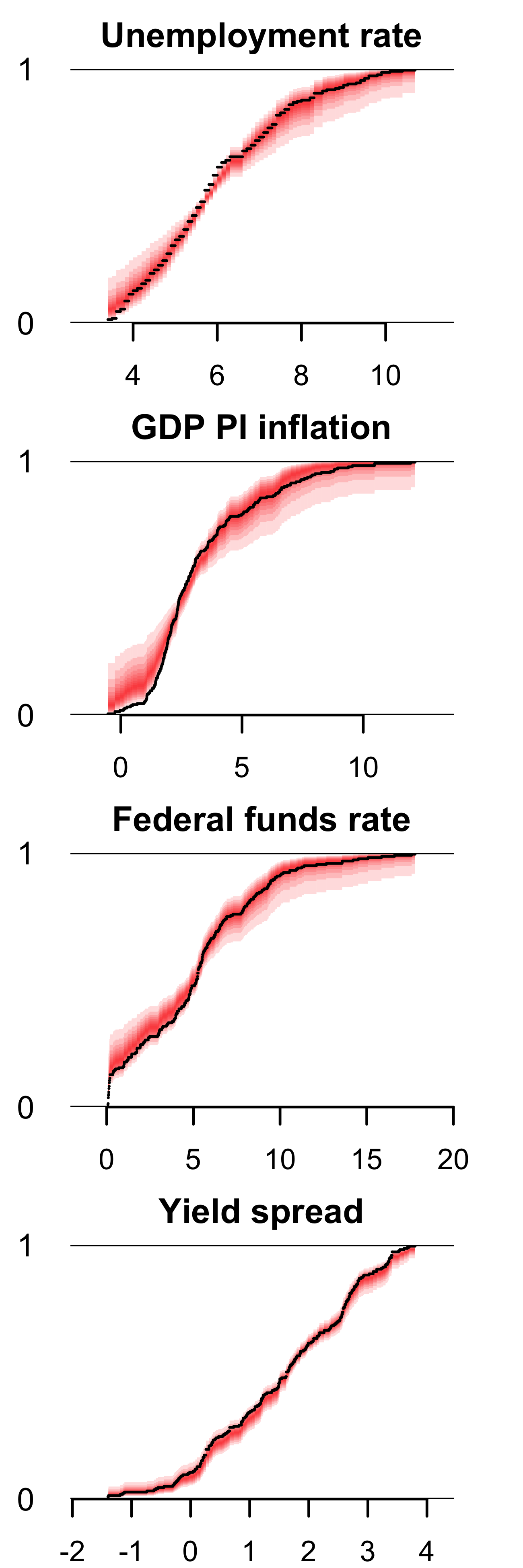}
    \includegraphics[width = 0.57\textwidth]{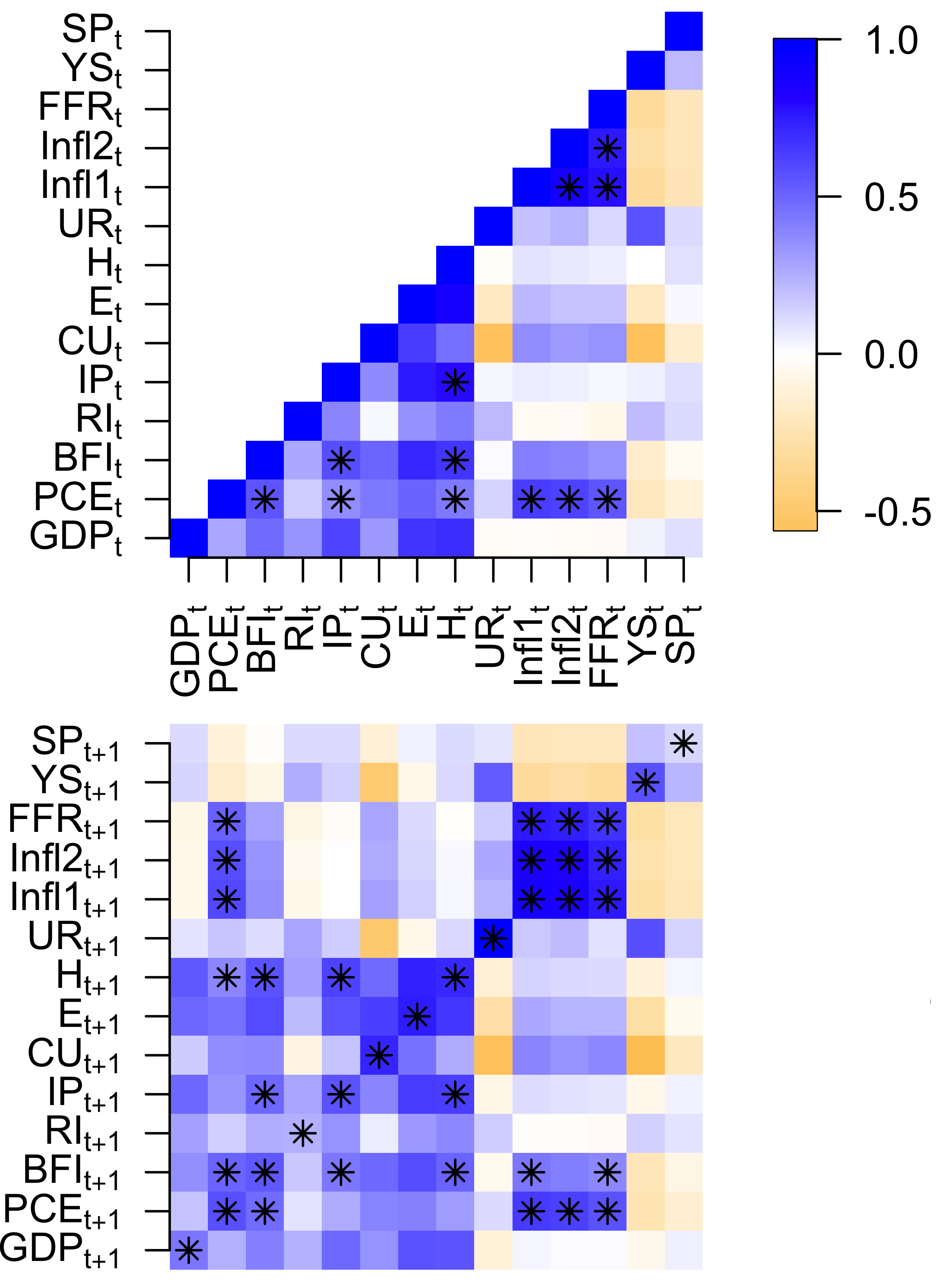}
    \caption{\small \textbf{Left}: the margin adjustment posterior (pointwise quantile bands of MCMC draws) and the raw ECDF for select macro variables. \textbf{Right}: posterior mean estimate of $\text{corr}(\Bz_t)$ (top) and $\text{corr}(\Bz_{t+1}\com\Bz_t)$ (bottom). A ``*'' indicates that the 95\% (with Bonferroni correction) highest posterior density interval excludes zero. Abbreviations are in Table~\ref{tab:macro}.}
    \label{fig:macro_params}
\end{figure}

Table~\ref{tab:macro_results} displays the one-step ahead forecasting results, averaged over the period 1985Q1 ($t_0$) to 2014Q1 ($T$). Across all variables, the DGFC delivers the smallest forecasting intervals---often by a wide margin---with close to nominal (95\%) coverage. By contrast, the DLM and BVAR are excessively conservative. Similarly favorable results are observed for density forecasting, with CRPS preferring DGFC for the majority of variables. Point forecasting performance tends to favor the DLM, although the DGFC remains competitive. These results suggest that the DGFC may be most useful for interval and density forecasting, both of which are essential for probabilistic uncertainty quantification. 

\setlength{\tabcolsep}{6pt}
\begin{table}[h] \small
    \centering
    \begin{tabular}{l  rrrr r  rrrr}
          & \textbf{MSE}  &  \textbf{COV}  &  \textbf{SIZE}  &  \textbf{CRPS}  &  &  \textbf{MSE}  &  \textbf{COV}  &  \textbf{SIZE}  &  \textbf{CRPS} \\
    \cline{2-5}\cline{7-10}
     & \multicolumn{4}{|c|}{Real GDP growth}  &  &  \multicolumn{4}{|c|}{PCE growth}\\\cline{2-5}\cline{7-10}
     
DGFC  &  6.31 & 0.97 & \textbf{11.32} & \textbf{1.37} &  & 6.75 & 0.90 & \textbf{9.29} & 1.44\\
BVAR  &  6.92 & 1.00 & 96.57 & 5.68 &  & 6.06 & 1.00 & 80.16 & 4.70\\
DLM  &  \textbf{4.55} & 0.99 & 16.37 & 1.39 &  & \textbf{4.30} & 1.00 & 13.45 & \textbf{1.26}\\\cline{2-5}\cline{7-10}

     & \multicolumn{4}{|c|}{BFI growth}  &  &  \multicolumn{4}{|c|}{Residential investment growth}\\\cline{2-5}\cline{7-10}
     
DGFC  & \textbf{ 67.97} & 0.93 & \textbf{30.37} & \textbf{4.33} &  & 251.07 & 0.96 & \textbf{68.98} & 8.27\\
BVAR  &  84.16 & 1.00 & 241.03 & 14.64 &  & 240.97 & 1.00 & 792.96 & 47.02\\
DLM  &  99.70 & 0.97 & 39.82 & 5.30 &  & \textbf{173.44} & 0.99 & 84.23 & \textbf{8.09}\\\cline{2-5}\cline{7-10}

     & \multicolumn{4}{|c|}{Industrial production growth}  &   & \multicolumn{4}{|c|}{Capacity utilization}\\\cline{2-5}\cline{7-10}
     
DGFC  &  \textbf{17.58} & 0.95 & \textbf{17.49} & \textbf{2.27} &  & 4.50 & 0.94 & \textbf{7.69} & \textbf{1.09}\\
BVAR  &  19.70 & 1.00 & 135.12 & 8.12 &  & \textbf{1.23} & 1.00 & 27.68 & 1.68\\
DLM  &  21.82 & 0.98 & 28.32 & 2.77 &  & 6.34 & 0.99 & 20.89 & 1.70\\\cline{2-5}\cline{7-10}

     & \multicolumn{4}{|c|}{Payroll employment growth}  &   & \multicolumn{4}{|c|}{Growth in aggregate hours}\\\cline{2-5}\cline{7-10}
     
DGFC  &  1.39 & 0.97 & \textbf{5.72} & \textbf{0.67} &  & \textbf{3.33} & 0.97 & \textbf{8.99} & \textbf{1.03}\\
BVAR  &  \textbf{1.06} & 1.00 & 31.11 & 1.85 &  & 4.74 & 1.00 & 70.56 & 4.19\\
DLM  &  1.35 & 1.00 & 10.90 & 0.83 &  & 5.03 & 0.99 & 16.15 & 1.41\\\cline{2-5}\cline{7-10}

     & \multicolumn{4}{|c|}{Unemployment rate}  &   & \multicolumn{4}{|c|}{GDP PI inflation}\\\cline{2-5}\cline{7-10}
     
DGFC  &  0.44 & 0.92 & \textbf{2.10} & 0.34 &  & 1.58 & 0.95 & \textbf{4.81} & \textbf{0.66}\\
BVAR  &  0.04 & 1.00 & 5.48 & \textbf{0.33} &  & 1.29 & 1.00 & 19.61 & 1.16\\
DLM  &  \textbf{0.23} & 1.00 & 8.48 & 0.55 &  & \textbf{0.98} & 1.00 & 11.76 & 0.83\\\cline{2-5}\cline{7-10}

     & \multicolumn{4}{|c|}{PCE inflation}  &  &  \multicolumn{4}{|c|}{Federal funds rate}\\\cline{2-5}\cline{7-10}
     
DGFC  &  3.46 & 0.91 & \textbf{5.98} & \textbf{0.94} &  & 3.06 & 0.95 & \textbf{6.12} & 0.94\\
BVAR  &  3.01 & 1.00 & 21.35 & 1.37 &  & 0.58 & 1.00 & 15.72 & 0.93\\
DLM  &  \textbf{2.36} & 0.99 & 12.05 & 0.98 &  & \textbf{0.77} & 1.00 & 9.17 & \textbf{0.68}\\\cline{2-5}\cline{7-10}

     & \multicolumn{4}{|c|}{Yield spread}  &  &  \multicolumn{4}{|c|}{Real S\&P 500 growth}\\\cline{2-5}\cline{7-10}
     
DGFC  &  0.79 & 0.93 & \textbf{3.13} & \textbf{0.51 }&  & 756.71 & 0.91 & \textbf{83.93} & 14.33\\
BVAR  &  \textbf{0.26} & 1.00 & 10.48 & 0.60 &  & 939.14 & 1.00 & 1686.73 & 100.5\\
DLM  &  0.35 & 1.00 & 7.32 & \textbf{0.51} &  & \textbf{698.01} & 0.97 & 99.12 & \textbf{13.93}\\
    \end{tabular}
    \caption{\small One-quarter-ahead forecast metrics for the macroeconomic data in Figure~\ref{fig:crime}, averaged over the period 1985Q1 to 2014Q1. 
    \textbf{MSE} is the mean squared error of the posterior predictive median (smaller is better); \textbf{COV}, \textbf{SIZE}, and \textbf{CRPS} are defined in Table~\ref{tab:crime}. The best (smallest) values of MSE, SIZE, and CRPS are bolded for each variable. 
    The proposed method (DGFC) consistently offers superior density and interval forecasts and competitive point forecasts.}
    \label{tab:macro_results}
\end{table}

\section{Conclusion}\label{sec:conclusion}

In this article, we proposed a posterior approximation strategy MTS copula models and applied it to a Gaussian copula based on a dynamic factor model in order to forecast MTS with a heterogeneous mix of distributional features and data types across variables. Our novel estimation scheme allows us to circumvent the usual limitations of MTS copula models by producing full Bayesian uncertainty quantification for all model unobservables without requiring variable-specific tuning or unwieldy posterior samplers. 

We demonstrated the efficacy of our approach in theoretical analyses and simulation studies, and in real data applications to crime counts and macroeconomic aggregates. Across these explorations, we confronted our method with a wide range of data, both real and simulated: cross-sections of all counts, cross-sections of all continuous variables, mixed cross-sections, variables with skew, heavy tails, multiple modes, zero-inflation. We also toured various forms of time series dynamics: autoregressive, moving average, linear, nonlinear, conditionally heteroskedastic, as well as real data that may not even be stationary. In all cases, our proposed approach worked exceptionally well without any user input or tuning: the results were generated using the same default priors and posterior sampling algorithm. Any variable- or application-specific data types and distributional features were automatically learned 
without any ad hoc interventions by the analyst. As such, we offer our approach as a versatile, general-purpose utility for MTS forecasting that works well across of range of applications with minimal user-intensive tuning. 

There are several promising avenues for future work. First, our approach may be extended for binary and categorical variables,  for instance by replacing the rank posterior approximation (Section~\ref{sec:erl}) with a rank-probit posterior \citep{fk2022aoas}. Second, we aim to incorporate nonstationary dynamics in \eqref{eq:ztrans} with a more general link \eqref{eq:link}, which would require careful modifications of the margin adjustment (Section~\ref{sec:ma}). Third, our rank posterior approximation uses 
ranks instead of observed data values, which may boost robustness---and thus combat a main limitation of Gaussian copulas. 
Finally, our broader framework applies for general MTS copula models; it would be interesting to explore the rich collection of non-Gaussian copulas. 



\bibliographystyle{ecta}
\bibliography{main.bib}


\newpage
\setcounter{page}{1}
{\, } \vspace{5mm}
\begin{center}
\Large 
    {Supplement to ``A dynamic copula model for probabilistic forecasting of non-Gaussian multivariate time series''}
\end{center}

\appendix 

\section{Theoretical background}

\subsection{The rank posterior}\label{app:rank}

The true posterior and the rank posterior are related as follows:
\begin{align*}
    p(\latent_{1:T}\com\Btheta\given\By_{1:T})
    &=
    p\{\latent_{1:T}\com\Btheta\given\By_{1:T}\com \latent_{1:T}\in\dcal(\By_{1:T})\}\\
    &=
        \frac{
    p\{\By_{1:T}\com\Bz_{1:T}\in\dcal(\By_{1:T})\given\Bz_{1:T}\com\Btheta\} \
    p(\Bz_{1:T}\com\Btheta)
    }
    {
    p\{\By_{1:T}\com\Bz_{1:T}\in\dcal(\By_{1:T})\}
    }\\
    &=        \frac{
    {p\{\By_{1:T}\given\Bz_{1:T}\in\dcal(\By_{1:T})\com\Bz_{1:T}\com\Btheta\}}
        {
        p\{\Bz_{1:T}\in\dcal(\By_{1:T})\given\Bz_{1:T}\com\Btheta\}\ 
        p(\Bz_{1:T}\com\Btheta)
        }
    }
    {
    {p\{\By_{1:T}\given\Bz_{1:T}\in\dcal(\By_{1:T})\}}
        {p\{\Bz_{1:T}\in\dcal(\By_{1:T})\}}
    }\\
    &=
    \frac{
    {
        p\{\Bz_{1:T}\in\dcal(\By_{1:T})\given\Bz_{1:T}\com\Btheta\} \
        p(\Bz_{1:T}\com\Btheta)
        }
    }
    {
    {p\{\Bz_{1:T}\in\dcal(\By_{1:T})\}}
    }
    \,\times \,
    \frac{
    {p\{\By_{1:T}\given\Bz_{1:T}\in\dcal(\By_{1:T})\com\Bz_{1:T}\com\Btheta\}}
    }
    {
    {p\{\By_{1:T}\given\Bz_{1:T}\in\dcal(\By_{1:T})\}}
    }\\
    &= p\{\Bz_{1:T}\com\Btheta\given\Bz_{1:T}\in\dcal(\By_{1:T})\}
    \,\times \,
    \frac{
    p\{\By_{1:T}\given\Bz_{1:T}\in\dcal(\By_{1:T})\com\Bz_{1:T}\com\Btheta\}
    }
    {
    p\{\By_{1:T}\given\Bz_{1:T}\in\dcal(\By_{1:T})\}
    }.
\end{align*}
The first line follows from the fact that, given $\By_{1:T}$, it must already be the case that $\Bz_{1:T}\in\dcal(\By_{1:T})$, so the added truncation is redundant. The second line is Bayes' theorem, the third line applies a marginal-conditional decomposition, the fourth line rearranges terms, and the final line applies Bayes' theorem to the first term. 

\subsection{Proof of Theorem~\ref{thm:ma}}\label{app:proof_ma}    

We first state a technical lemma:

\begin{lemma}\label{lem:lemma}
    Let $U_t$ be a strictly stationary ergodic process with continuous $H(x)=P(U_1\leq x)$. Then for any $x$, $M_T=\max_{t=1,...,T}\{U_t:U_t\leq x\}\cas x$.
\end{lemma}

    




The main result follows:

\begin{proof}[Proof of Theorem~\ref{thm:ma}]
    Since $Z_t$ is strictly stationary and ergodic, $U_t=G(Z_t)$ is also strictly stationary and ergodic, and $U_t\sim\text{Unif}(0\com 1)$ for all $t$ by the probability integral transform. Noting that 
    \begin{align*}
        \tilde{F}_T(x)
        &=
        \max_{t=1,...,T}\{G(Z_t):Y_t\leq x\}
        \\
        &=
        \max_{t=1,...,T}\{G(Z_t):F^{-1}\circ G(Z_t)\leq x\}
        \\
        &=
        \max_{t=1,...,T}\{G(Z_t):G(Z_t)\leq F(x)\}
        \\
        &=
        \max_{t=1,...,T}\{U_t:U_t\leq F(x)\},
    \end{align*}
    Lemma~\ref{lem:lemma} tell us that $\tilde{F}_T(x)\cas F(x)$.
\end{proof}

\subsection{Identification in the VAR copula}\label{app:id}

Our posterior consistency theory in Section~\ref{sec:theory} applies when the true data-generating process is a VAR copula:
\begin{align*}
    \unscaled_t&=\BG\unscaled_{t-1}+\Bepsilon_t,&&\Bepsilon_t\iid\text{N}_n(\Bzero\com\BSigma)\\
    \latent_t&=\BD_0^{-1/2}\unscaled_t &&\BD_0=\text{diag}(\text{var}(\unscaled_t))\\
    y_{t,i}&=F_i^{-1}\circ\Phi(z_{t,i}).
\end{align*}
In general, the unrestricted VAR(1) parameters $\BG$ and $\BSigma$ are not identified, but the normalized variate $\latent_t$ follows a unit-variance VAR(1), and its parameters $\tilde{\BG}$ and $\tilde{\BSigma}$ are identified. These are what we refer to in (\ref{eq:simple_latent}) and Theorem~\ref{thm:doob1}. To see this, note that the stationary process $\unscaled_t$ will have 
$$
\begin{bmatrix}
    \unscaled_{t-1} \\
    \unscaled_t
\end{bmatrix}
\sim 
\text{N}_{2n}
\left(
\Bzero 
\com 
\BOmega 
= 
\begin{bmatrix}
    \BOmega_0 & \BOmega_1^\tr \\
    \BOmega_1 & \BOmega_0
\end{bmatrix}
\right)
,
$$
where $\text{cov}(\unscaled_t)=\BOmega_0$ and $\text{cov}(\unscaled_t\com\unscaled_{t-1})=\BOmega_1$, and we can calculate these with $\text{vec}(\BOmega_0)=(\BI_{n^2}-\BG\otimes\BG)^{-1}\text{vec}(\BSigma)$ and $\BOmega_1=\BG\BOmega_0$ \citep{lutkepohl2006book}. We see then that $\BD_0=\text{diag}(\BOmega_0)$, and the normalized process $\latent_t$ has 
\begin{equation}\label{eq:apprescale}
\begin{bmatrix}
    \latent_{t-1} \\
    \latent_t
\end{bmatrix}
=
\BI_2
\otimes 
\BD_0^{-1/2}
\begin{bmatrix}
    \unscaled_{t-1} \\
    \unscaled_t
\end{bmatrix}
\sim 
\text{N}_{2n}
\left(
\Bzero 
\com 
\BC 
= 
\begin{bmatrix}
    \BC_0 & \BC_1^\tr \\
    \BC_1 & \BC_0
\end{bmatrix}
\right)
,
\end{equation}
where $\BC_1=\BD_0^{-1/2}\BOmega_1\BD_0^{-1/2}$, and $\BC_0=\BD_0^{-1/2}\BOmega_0\BD_0^{-1/2}$ is a correlation matrix. As such, we see that $\latent_t$ is a unit-variance VAR(1) process, and we can calculate its implied parameters by inverting the relationship $(\BG\com\BSigma)\mapsto(\BOmega_0\com\BOmega_1)$ from above:
\begin{align}
    \tilde{\BG}&=\BC_1\BC_0^{-1}\label{eq:yw1}\\
    \text{vec}(\tilde{\BSigma})&=(\BI_{n^2}-\tilde{\BG}\otimes\tilde{\BG})\text{vec}(\BC_0).\label{eq:yw2}
\end{align}
So, integrating over $\unscaled_t$, we are left with a specification like (\ref{eq:simple_latent}, \ref{eq:stationary_link}), where the parameters $(\tilde{\BG}\com\tilde{\BSigma})$ necessarily imply that the stationary distribution has unit variance, ensuring that the parameters are identifiable. Two final points on this:
\begin{itemize}
    \item The assumptions of Theorem~\ref{thm:doob1} refer to a prior distribution for $(\tilde{\BG}\com\tilde{\BSigma})$ supported only on the set $\Theta\subseteq\RR^{n\times n}\times\text{SPD}_n$ where unit variance holds. We see now how such a prior could be constructed. For any prior on the unrestricted VAR parameters $(\BG\com\BSigma)$, we can draw $\BG\com\BSigma\sim p(\BG\com\BSigma)$ and then transform $(\BG\com\BSigma)\mapsto(\BOmega_0\com\BOmega_1)\mapsto(\BC_0\com\BC_1)\mapsto(\tilde{\BG}\com\tilde{\BSigma})$ according to the formulas above. This scheme generates iid draws from an induced prior for $(\tilde{\BG}\com\tilde{\BSigma})$, where the restriction holds. This is the time series analog of inducing a prior on a correlation matrix $\BC$ by placing a prior on the associated covariance matrix $\BV\sim p(\BV)$ and then normalizing $C_{i,j}=V_{i,j}/\sqrt{V_{i,i}V_{j,j}}$. See \cite{hoff2007aoas} for discussion.
    \item Similarly, Algorithm~\ref{alg:gibbs_sampler} produces MCMC draws from the posterior for the unrestricted VAR parameters. As a post-processing step, we can compute $(\BG\com\BSigma)\mapsto(\BOmega_0\com\BOmega_1)\mapsto(\BC_0\com\BC_1)\mapsto(\tilde{\BG}\com\tilde{\BSigma})$ on a draw-by-draw basis to approximate the posterior for the identifiable combination of parameters. This is what we do in Section~\ref{app:sim}.
\end{itemize}

\subsection{Proof of Theorem~\ref{thm:doob1}}\label{app:proof_doob1}

Similar to \cite{mdcl2013jasa}, we prove the result by appealing to Doob's theorem:
\begin{lemma}\label{lem:doob} (\textbf{Doob's theorem as stated by \citealt{gg2009jspi}}) Let $\Bx_{1:t}$ be observations whose distribution depends on a parameter $\Btheta$, both of which take values in a Polish space. Let $\Btheta\sim\Pi$, let $\Bx_{1:t}\given\Btheta\sim P^{(t)}_{\Btheta}$, let $\xcal_t$ be the $\sigma$-algebra generated by $\Bx_{1:t}$, and let $\xcal_\infty=\sigma\left(\cup_{t=1}^\infty\xcal_t\right)$. If there exists a $\xcal_\infty$-measurable function $f$ such that $\Btheta=f(\omega)$ a.e.\ $[\Pi\times P_{\Btheta}^{(\infty)}]$, then the posterior is strongly consistent at $\Btheta$ for almost every $\Btheta$ $[\Pi]$.
\end{lemma}
We know from \cite{lpw2007et} that this theorem still holds for stationary models like ours, and so our proof consists of showing that the conditions of Lemma~\ref{lem:doob} are satisfied. This amounts to showing that there is a consistent estimator of the VAR parameters $(\tilde{\BG}\com\tilde{\BSigma})$ that is measureable with respect to the $\sigma$-field generated by the sequence $\{\dcal(\By_{1:t})\}_{t=1}^\infty$:
\begin{proof}[Proof of Theorem~\ref{thm:doob1}]
To begin, note that Theorem~\ref{thm:doob1} requires that the $F_i$ be continuous. In this case, the information contained in $\{\dcal(\By_{1:t})\}_{t=1}^\infty$ is equivalent to the information contained in the ranks $\{r(\By_{1:t})\}_{t=1}^\infty$ \citep{hoff2007aoas}. So henceforth, we work exclusively with the ranks. To that end, let $\rcal_t$ denote the $\sigma$-field generated by $\{r(\By_{1:j})\}_{j=1}^t$, and $\rcal_\infty=\sigma\left(\cup_{t=1}^\infty \rcal_t\right).$ Next, let $u_{t,i}=F_i(y_{t,i})$, and collect them in $\Bu_{t}=[u_{t,1}\,\cdots\,u_{t,n}]^\tr$. Furthermore, let $r_{t,i}^{(T)}$ denote the rank of observation $y_{t,i}$ in the list $y_{1:T,i}$, define $\hat{u}_{t,i}^{(T)}=r_{t,i}^{(T)}/(T+1)$, and collect these in $\hat{\Bu}^{(T)}_t=[\hat{u}_{t,1}^{(T)}\,\cdots\,\hat{u}_{t,n}^{(T)}]^\tr$. We know by the ergodic theorem that $\hat{u}_{t,i}^{(T)}\cas u_{t,i}$ as $T\to\infty$, and so $\hat{\Bu}^{(T)}_t\cas \Bu_t$. This means that the $\Bu_t$ are $\rcal_\infty$-measurable. Armed with these, a consistent, rank-based estimator of $\Btheta$ can be constructed along the lines of \cite{fhp2023joe}, which requires that the marginal distributions be continuous. The matrix $\BC$ in (\ref{eq:apprescale}) can be estimated with $\hat{\BC}_T=\sin(\pi\hat{\BT}_T/2)$, where 
$$
\hat{\BT}_T
=
\frac{2}{(T-1)(T-2)}\sum\limits_{i=1}^{T-1}\sum\limits_{j=i+1}^{T-1}\text{sign}\begin{bmatrix}
    \Bu_i-\Bu_j\\ \Bu_{i+1}-\Bu_{j+1}
\end{bmatrix}
\text{sign}\begin{bmatrix}
    \Bu_i-\Bu_j\\ \Bu_{i+1}-\Bu_{j+1}
\end{bmatrix}^\tr
.
$$
Given those estimates, we can calculate $\hat{\Btheta}_T=\{\hat{\BG}_T\com\hat{\BSigma}_T\}$ by applying (\ref{eq:yw1}, \ref{eq:yw2}) using the contents of $\hat{\BC}_T$. $\hat{\Btheta}_T$ is a function of the $\Bu_{t}$, which are $\rcal_\infty$-measurable, and so we have the result.
\end{proof}

\subsection{Proof of Theorem~\ref{thm:doob2}} \label{app:proof_doob2}

This closely follows \cite{fk2024jmlr}:

\begin{proof}[Proof of Theorem~\ref{thm:doob2}]
Fix arbitrary $x\in\RR$ and $i=1\com2\com...\com n$. We are done if we can produce a consistent estimator of $F_i(x)$ that is measureable with respect to the $\sigma$-field generated by $\{\dcal(\By_{1:t})\}_{t=1}^\infty$. Then the result follows again from Doob's theorem (Lemma~\ref{lem:doob}). To begin, consider the marginal rank event
$$
\dcal(y_{1:t,i})
=
\{
z_{1:t,i}\in\RR^t:y_{k,i}<y_{k',i}\implies z_{k,i}<z_{k',i}
\}.
$$
The $\sigma$-field generated by $\{\dcal(y_{1:t,i})\}_{t=1}^T$ is a sub $\sigma$-field of the $\sigma$-field generated by $\{\dcal(\By_{1:t})\}_{t=1}^T$. Functions measureable with respect to the first will therefore be measureable with respect to the second, and it suffices to work with the first.

Since we always have $z_{1:t,i}\in\dcal(y_{1:t,i})$, then $\tilde{F}_i(x)=\max_t\{\Phi(z_{t,i}):y_{t,i}\leq x\}$ is measurable with respect to the $\sigma$-field generated by $\{\dcal(y_{1:t,i})\}_{t=1}^T$. Theorem~\ref{thm:ma} says that $\tilde{F}_i(x)$ consistently estimates $F_i(x)$ in the context of our Gaussian VAR copula, so $F_i(x)$ is measurable with respect to $\{\dcal(y_{1:t,i})\}_{t=1}^\infty$, and hence $\{\dcal(\By_{1:t})\}_{t=1}^\infty$.

\end{proof}

\subsection{Simulation}\label{app:sim}

Lastly, we present simulation results to see Theorems~\ref{thm:doob1} in action. We simulate datasets of size $T=25\com50\com100\com200\com...\com3200$, refit our model for each sample size, and observe in Figure~\ref{fig:more_simulation_recovery} how the pseudo-posterior (visualized with box plots of the MCMC draws) concentrates around the ground truth values of the copula parameters. We consider synthetic data from (\ref{eq:stationary_link},\ref{eq:simple_latent}), which is the special case of our factor model with $n=k$, $\BLambda=\BI_n$, and $\BV=\Bzero$. So we can fit this model with the special case of Algorithm~\ref{alg:gibbs_sampler} that skips the steps for those unobservables (and all associated hyperparameters). Recalling the discussion in Section~\ref{app:id}, this algorithm produces draws $\{\BG^{(m)}\com\BSigma^{(m)}\}$ that we subsequently transform to $\{\tilde{\BG}^{(m)}\com\tilde{\BSigma}^{(m)}\}$ for the sake of identification.

\bigskip 

\noindent The simulation settings are:

\begin{itemize}
    \item (\textbf{Ground truth}) Fix $n=2$, $\text{vec}(\BG_0)\sim\text{N}_{n^2}(\Bzero\com 0.1\cdot\BI_{n^2})$ (rejecting draws until stationarity), $\BSigma_0\sim\text{IW}_n(n+1\com\BI_n)$, $F_1=\text{Gamma}(1\com 1)$, $F_2=\text{Skew-}t(3\com0\com1\com2)$. The identifiable parameters $(\tilde{\BG}_0\com\tilde{\BSigma}_0)$ of the implied VAR(1) for $\tilde{\Bz}_t$ are then computed as in Appendix~\ref{app:id}, and these are the ground truth values indicated with a red line in Figure~\ref{fig:more_simulation_recovery};
    \item (\textbf{Prior hyperparameters}) Fix $d_0=n+1$, $\BPsi_0=\BI_n$, $\overline{\BG}_0=\Bzero$, $\BO^{-1}_0=\BI_n$;
    \item (\textbf{MCMC settings}) For each rerun of the MCMC, we generate 5,000 posterior draws by taking 35,000 draws in total, discarding the first 5,000 as burn-in, and then thinning to every 6th draw.
\end{itemize}



\begin{figure}
    \centering
    \includegraphics[width = 0.95\textwidth]{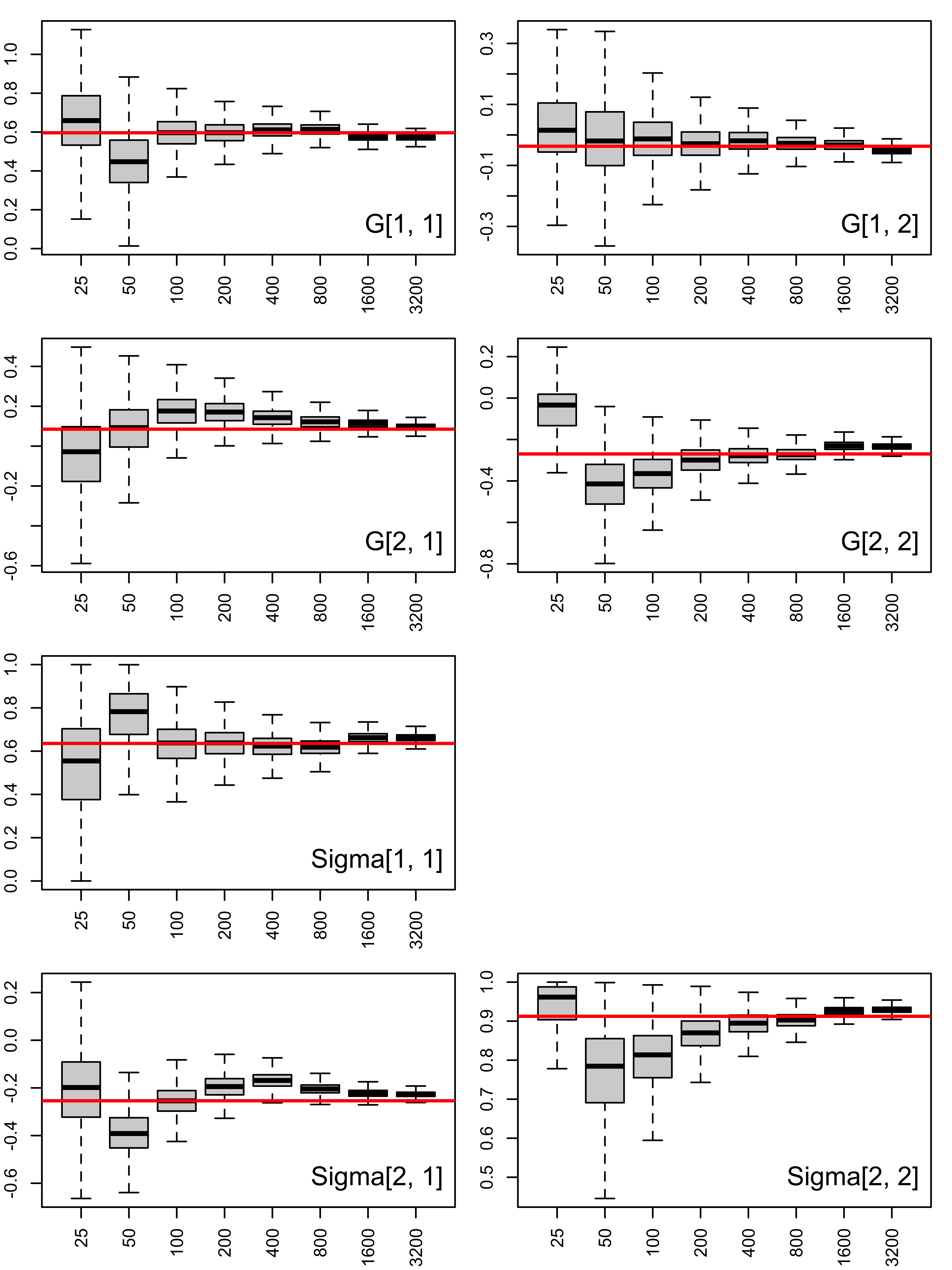}
    \caption{As we fit our model to larger simulated data ($T=25\com 50\com 100\com ...\com3200$), the pseudo-posterior concentrates around the ground truth values of the copula parameters (Theorem~\ref{thm:doob1}).}
    \label{fig:more_simulation_recovery}
\end{figure}

\section{The dependence structure of the DGFC}\label{sec:dgfc_dependence}

Our dynamic factor model (\ref{eq:model_xz} - \ref{eq:model_factor}) can be rewritten as $\Bz_{1:T}\sim\text{N}_{Tn}(\Bzero\com\BC_{\Btheta})$, and the structure of the correlation matrix $\BC_{\Btheta}$ encodes the cross-sectional and serial dependencies that are passed to $\By_t$ through the copula link functions. We now make this structure explicit. To begin, assume that the matrix $\BG$ has eigenvalues inside the unit circle, which we enforce during inference. This implies that the VAR(1) in (\ref{eq:model_factor}) defines a stationary, mean-zero Gaussian process. As such, we have that $\factor=[\factor_1^\tr\,\factor_2^\tr\,\cdots\,\factor_T^\tr]^\tr\sim\text{N}_{Tk}(\Bzero\com\BGamma)$ with 
$$
\BGamma
=
\begin{bmatrix}
\BGamma_0 & \BGamma_1^\tr  & \cdots & \BGamma_{T-2}^\tr & \BGamma_{T-1}^\tr\\
\BGamma_1 & \BGamma_0  & \cdots & \BGamma_{T-3}^\tr & \BGamma_{T-2}^\tr\\
\vdots & \vdots  & \ddots & \vdots & \vdots\\
\BGamma_{T-2} & \BGamma_{T-3}  & \cdots & \BGamma_0 & \BGamma_1^\tr \\
\BGamma_{T-1} & \BGamma_{T-2}  & \cdots & \BGamma_1 & \BGamma_0
\end{bmatrix},
$$
and $\BGamma_h = \text{cov}(\Beta_t\com\Beta_{t- h})$ for any $t$. From (\ref{eq:model_xz}) we see that $\unscaled\given\Beta\sim\text{N}_{Tn}\left((\BI_T\otimes\BLambda)\Beta\com \BI_T\otimes \BV\right)$, and so marginally, $\unscaled\sim\text{N}_{Tn}\left(\Bzero\com \BOmega\right)$ with $\BOmega=(\BI_T\otimes\BLambda)\BGamma(\BI_T\otimes\BLambda^\tr)+\BI_T\otimes \BV$. This implies a familiar-looking fact about Gaussian factor models, that $\unscaled_t\sim\text{N}_n(\Bzero\com \BOmega_0)$ with $\BOmega_0=\BLambda\BGamma_0\BLambda^\tr+\BV$ for all $t=1\com...\com T$, but in addition the $\unscaled_t$ are serially correlated according to the contents of $\BOmega$. We now also see that  $\BD_0$ in (\ref{eq:model_factor}) refers to the diagonal entries of $\BOmega_0$, which are the same for all $t$. Lastly, (\ref{eq:model_factor}) implies that $\latent=(\BI_T\otimes\BD_0^{-1/2})\unscaled\sim\text{N}_{Tn}(\Bzero\com\copulacov_{\Btheta})$ with correlation matrix $\copulacov_{\Btheta}=(\BI_T\otimes\BD_0^{-1/2})\BOmega(\BI_T\otimes\BD_0^{-1/2})$. In particular, $\Bz_t\sim\text{N}_{n}(\Bzero\com \BC_0)$, where $\text{corr}(\Bz_t)=\text{cov}(\Bz_t)=\BC_0=\BD_0^{-1/2}\BOmega_0\BD_0^{-1/2}$ for all $t$. Posterior means of this functional are displayed in Figures~\ref{fig:crime_params} and \ref{fig:macro_params}. These figures also visualize $\text{corr}(\Bz_{t+1}\com \Bz_t)$, which is an off-diagonal block of $\copulacov_{\Btheta}$.

\section{Prior details and posterior computation}

We choose a conditionally conjugate prior $\BG^\tr\com\BSigma\sim\text{MNIW}_{k,k}(d_0\com\BPsi_0\com\overline{\BG}_0^\tr\com\BO^{-1}_0)$ for the VAR parameters in (\ref{eq:model_factor}), and to ensure that this process is stationary, we truncate the prior (and consequently the posterior) so that $\BG$ always has eigenvalues within the unit circle. For the variances $\Bv=[v_1\,v_2\,\cdots\,v_n]^\tr$, we take $v_i^{-1}\sim\text{iid}\,\text{Gamma}(\alpha_0\com\beta_0)$. For the factor loadings $\BLambda=[\lambda_{i,l}]$, we use the multiplicative gamma process prior of \cite{bd2011bka}:
\begin{align}
    \lambda_{i,l}&\indep\text{N}(0\com \phi^{-1}_{i,l}\tau_l^{-1})\\
    \phi_{i,l}&\iid\,\text{Gamma}(\nu_0/2\com\nu_0/2)\\
    \tau_l&=\prod_{h=1}^l\delta_h \quad \begin{array}{l}
        \delta_1\sim\text{Gamma}(a_0\com 1)  \\
          \delta_h\iid\text{Gamma}(b_0\com 1).
    \end{array}
\end{align}
The $\phi_{i,l}$ are local scale parameters, the $\tau_l$ are global scale parameters, and the main advantage of this ordered shrinkage prior is to reduce sensitivity to the choice of factor dimension $k$. We collect these hyperparameters in $\BPhi=[\phi_{i,l}]$ and $\Btau=[\tau_1\,\tau_2\,\cdots\,\tau_k]^\tr$, and we include them in the set of copula parameters $\Btheta=\{\BG\com\BSigma\com\BLambda\com\Bv\com\BPhi\com\Btau\}$ that we estimate. In each of the simulations and applications, we use the same default prior hyperparameters: $k=\lceil0.7 n\rceil$, $d_0=k+1$, $\BPsi_0=\BI_k$, $\overline{\BG}_0=\Bzero$, $\BO_0=\BI_k$, $\alpha_0=1$, $\beta_0=0.3$, $\nu_0=3$, $a_0=2$, $b_0=3$.

\begin{algorithm}
\caption{Gibbs sampling steps for $p\{\latent_{1:T}\com\Beta_{1:T}\com\BG\com\BSigma\com\BLambda\com\Bv\com\BPhi\com\Btau\given\latent_{1:T}\in\dcal(\By_{1:T})\}$}

\bigskip

\For{$p(\BG\com\BSigma\given \latent_{1:T}\com \Beta_{1:T}\com\BLambda\com\Bv\com\BPhi\com\Btau)$}{
Construct pseudo-data 
$\BY = \begin{bmatrix}
\Beta_2 & \Beta_3 & \cdots & \Beta_T
\end{bmatrix}^\tr$ and $\BX = \begin{bmatrix}
\Beta_1 & \Beta_2 & \cdots & \Beta_{T-1}
\end{bmatrix}^\tr
$\;

Compute $d_T=d_0+T-1$, $\BO_T=\BX^\tr\BX+\BO_0$, $\overline{\BG}_T^\tr=\BO_T^{-1}\left(\BX^\tr\BY+\BO_0\overline{\BG}_0^\tr\right)$, and 
$$
    \BPsi_T
    =
    \BPsi_0 
    +
    \left(\BY-\BX\overline{\BG}_T^\tr\right)^\tr\left(\BY-\BX\overline{\BG}_T^\tr\right)
    +
    \left(\overline{\BG}_T^\tr-\overline{\BG}_0^\tr\right)^\tr\BO_0\left(\overline{\BG}_T^\tr-\overline{\BG}_0^\tr\right).
$$
Draw $\BG^\tr\com\BSigma\sim \textrm{MNIW}_{k\com k}(d_T\com\BPsi_T\com\overline{\BG}_T^\tr\com\BO^{-1}_T)$ until $\BG$ has eigenvalues in the unit circle.
}

\bigskip

\For{$p(\Beta_{1:T}\given \latent_{1:T}\com\BG\com\BSigma\com\BLambda\com\Bv\com\BPhi\com\Btau)$}{
(\ref{eq:model_xz},\ref{eq:model_factor}) define a linear Gaussian state space system with initial condition $\Beta_1\sim\text{N}_k(\Bzero\com\BGamma_0)$ as in Section~\ref{sec:dgfc_dependence}, so if the $\unscaled_{1:T}$ are observed, the $\Beta_{1:T}$ can be drawn using an implementation of the Kalman simulation smoother: \cite{dk2002bka}, \cite{cj2009ijmmno}, etc. 
}

\bigskip

\For{$p(\BLambda\given \latent_{1:T}\com\Beta_{1:T}\com\BG\com\BSigma\com\Bv\com\BPhi\com\Btau)$}{
Set $\BN=\begin{bmatrix}
\Beta_1 & \Beta_2 & \cdots & \Beta_T
\end{bmatrix}^\tr$, $\Bz_i=[z_{1,i}\,z_{2,i}\,\cdots\,z_{T,i}]^\tr$, and $\BP_i^{-1}=\text{diag}\left(\phi_{i,1}\tau_1\com ...\com\phi_{i,k}\tau_k\right)$\;
Draw $\Blambda_i\sim\text{N}_k\left(v_i^{-1}(\BP_i^{-1}+v_i^{-1}\BN^\tr\BN)^{-1}\BN^\tr\Bz_i\com(\BP_i^{-1}+v_i^{-1}\BN^\tr\BN)^{-1}\right)$ for each $i=1\com2\com...\com n$\;
Set $\BLambda=[\Blambda_1\,\cdots\,\Blambda_n]^\tr$.
}

\bigskip

\For{$p(\Bv\given\latent_{1:T}\com\Beta_{1:T}\com\BG\com\BSigma\com\BLambda\com\BPhi\com\Btau)$}{
Draw $v_i^{-1}\sim\text{Gamma}\left(\alpha_0+T/2\com\beta_0+\sum_{t=1}^T(z_{t,i}-\Blambda_i^\tr\Beta_t)^2/2\right)$ for each $i=1\com2\com...\com n$\;
}

\bigskip

\For{$p(\BPhi\given \latent_{1:T}\com\Beta_{1:T}\com\BG\com\BSigma\com\BLambda\com\Bv\com\Btau)$}{
Draw $\phi_{i,l}\sim\text{Gamma}\left(\frac{\nu_0+1}{2}\com\frac{\nu_0+\tau_l\lambda_{i,l}^2}{2}\right)$ for each $i=1\com2\com...\com n$ and $l=1\com2\com...\com k$\;
}

\bigskip

\For{$p(\Btau\given \latent_{1:T}\com\Beta_{1:T}\com\BG\com\BSigma\com\BLambda\com\Bv\com\BPhi)$}{
Draw $\delta_1\sim\text{Gamma}\left(a_0+\frac{nk}{2}\com 1+\frac{1}{2}\sum_{l=1}^k\tau_l^{(1)}\sum_{i=1}^n\phi_{i,l}\lambda_{i,l}^2\right)$\;
Draw $\delta_s\sim\text{Gamma}\left(b_0+\frac{n(k-s+1)}{2}\com1+\frac{1}{2}\sum_{l=s}^k\tau_l^{(s)}\sum_{i=1}^n\phi_{i,l}\lambda_{i,l}^2\right)$, $\tau_l^{(s)}=\Pi_{h=1,h\neq s}^l\delta_h$ \;
Compute $\tau_l=\prod_{h=1}^l\delta_h$ \;
}

\bigskip

\For{$p\left(z_{t,i}\given\Beta_{1:T}\com\BG\com\BSigma\com\BLambda\com\Bv\com\BPhi\com\Btau\com\latent_{1:T}/\{z_{t,i}\}\com\latent_{1:T}\in\dcal(\By_{1:T})\right)$}{
Compute $x_{t,i}^L=\max_{t'}\{x_{t',i}:y_{t',i}<y_{t,i}\}$ and $x_{t,i}^U=\min_{t'}\{x_{t',i}:y_{t',i}>y_{t,i}\}$\;
Draw $x_{t,i}\sim\text{TN}\left(\sum\limits_{l=1}^k\lambda_{i,l}\eta_{t,l}\com v_i\com x_{t,i}^L\com x_{t,i}^U\right)$\;
Compute $z_{t,i}=x_{t,i}/\sqrt{\omega_i}$.
}

\bigskip

\label{alg:gibbs_sampler}
\end{algorithm}

\section{Additional crime count forecasting results}

Table~\ref{tab:3crime} displays the one-step-ahead forecasting metrics when we redo the exercise in Section~\ref{sec:crime} with a smaller cross-section of only $n=3$ variables (murder, attempted murder, and accessory). Comparing Table~\ref{tab:3crime} and Table~\ref{tab:crime}, we see that the performance of the DGFC remains stable as the dimension increases, while the performance of the competing models worsens. 

\begin{table}
    \centering
    \begin{tabular}{lrrrr}
           & \textbf{MAE} &\textbf{COV} &\textbf{SIZE} &\textbf{CRPS}\\\cline{2-5}
           & \multicolumn{4}{|c|}{Murder}  \\\cline{2-5}
DGFC      & \textbf{2.10}  &  0.96  &  11.52 &\textbf{1.52}\\
BVAR      &2.32  &  1.00  &  37.46 &2.55\\
Pois-DGLM &2.14  &  0.96  &  \textbf{10.54} &\textbf{1.52}\\
NB-DGLM   &2.24  &  1.00  &  31.33 &2.10\\\cline{2-5}
           & \multicolumn{4}{|c|}{Attempted murder}  \\\cline{2-5}
DGFC      &1.80  &  0.99  &  10.12 &1.28\\
BVAR      &2.33  &  1.00  &  50.55 &3.23\\
Pois-DGLM &\textbf{1.78}  &  0.98  &   \textbf{8.51} & \textbf{1.25}\\
NB-DGLM   &1.88  &  1.00  &  18.87 &1.46\\\cline{2-5}
           & \multicolumn{4}{|c|}{Accessory to murder}  \\\cline{2-5}
DGFC      &\textbf{0.40}  &  0.99  &   \textbf{2.70} &\textbf{0.30}\\
BVAR      &0.62  &  0.99  &   4.64 &0.44\\
Pois-DGLM &0.48  &  0.98  &   2.95 &0.36\\
NB-DGLM   &\textbf{0.40}  &  1.00  &   2.83 &\textbf{0.30}\\
    \end{tabular}
    \caption{One-step ahead forecasting results for a smaller cross-section of $n=3$ crime variables. Metrics are defined in Table~\ref{tab:crime}. The best (smallest) values of MAE, SIZE, and CRPS are bolded for each variable. The proposed method (DGFC) consistently provides highly competitive density, interval, and point forecasts, but with fewer gains than in the larger cross-section ($n=10$) case. 
    }
    \label{tab:3crime}
\end{table}

\section{Details of the macroeconomic forecasting exercise}

In the macroeconomic forecasting application, we use the same dataset as \cite{ccm2016jbes}. It includes fourteen quarterly macroeconomic time series from 1965Q1 to 2019Q4. The data are displayed in Figure~\ref{fig:full_macro}, and further details are given in Table~\ref{tab:macro}. Using real-time data vintages from the Federal Reserve Bank of Philadelphia’s Real-Time Data Set for Macroeconomists (RTDSM), we analyze the real-time, one-step-ahead forecasts that our competing models would have produced if they had been recursively refit each quarter from 1985:Q1 through 2013:Q4.

\setlength{\tabcolsep}{6pt}
\begin{table}
    \centering
    \begin{tabular}{llll}\hline
        Variable & Abbrv & FRED mnemonic & Transformation  \\\hline
        Real GDP & GDP & \href{https://fred.stlouisfed.org/series/GDPC1}{\texttt{GDPC1}} & $400\Delta\ln$  \\
        Personal consumption expenditures & PCE & \href{https://fred.stlouisfed.org/series/PCE}{\texttt{PCE}} & $400\Delta\ln$  \\
        Business fixed investment & BFI & \href{https://fred.stlouisfed.org/series/PNFI}{\texttt{PNFI}} & $400\Delta\ln$  \\
        Residential investment & RI & \href{https://fred.stlouisfed.org/series/PRFI}{\texttt{PRFI}} & $400\Delta\ln$  \\
        Industrial production & IP & \href{https://fred.stlouisfed.org/series/INDPRO}{\texttt{INDPRO}} & $400\Delta\ln$  \\
        Capacity utilization in manufacturing & CU & \href{https://fred.stlouisfed.org/series/CUMFNS}{\texttt{CUMFNS}} & --  \\
        Nonfarm payroll employment & E & \href{https://fred.stlouisfed.org/series/PAYEMS}{\texttt{PAYEMS}} & $400\Delta\ln$  \\
        Aggregate hours worked & H & \href{https://fred.stlouisfed.org/series/AWHI}{\texttt{AWHI}} & $400\Delta\ln$  \\
        Unemployment rate & UR & \href{https://fred.stlouisfed.org/series/UNRATE}{\texttt{UNRATE}} & --  \\
        GDP price index inflation & Infl1 & \href{https://fred.stlouisfed.org/series/GDPCTPI}{\texttt{GDPCTPI}} & $400\Delta\ln$  \\
        PCE inflation & Infl2 & \href{https://fred.stlouisfed.org/series/PCEPI}{\texttt{PCEPI}} & $400\Delta\ln$  \\
        Federal funds rate & FFR & \href{https://fred.stlouisfed.org/series/DFF}{\texttt{DFF}} & --  \\
        10-year, 3-month treasury yield spread & YS & \href{https://fred.stlouisfed.org/series/DGS10}{\texttt{DGS10}} - \href{https://fred.stlouisfed.org/series/DTB3}{\texttt{DTB3}} & --  \\
        Real S\&P 500 & SP & \href{https://fred.stlouisfed.org/series/SP500}{\texttt{SP500}} / \href{https://fred.stlouisfed.org/series/PCEPI}{\texttt{PCEPI}} & $400\Delta\ln$  \\\hline
    \end{tabular}
    \caption{Macroeconomic variables studied by \cite{ccm2016jbes}. Higher frequency variables are aggregated to the quarterly level by averaging. Transformed variables are in terms of annualized percent growth rates. Data are available from Federal Reserve Economic Data (FRED).}
    \label{tab:macro}
\end{table}

\begin{figure}
    \centering
    \includegraphics[width=\textwidth]{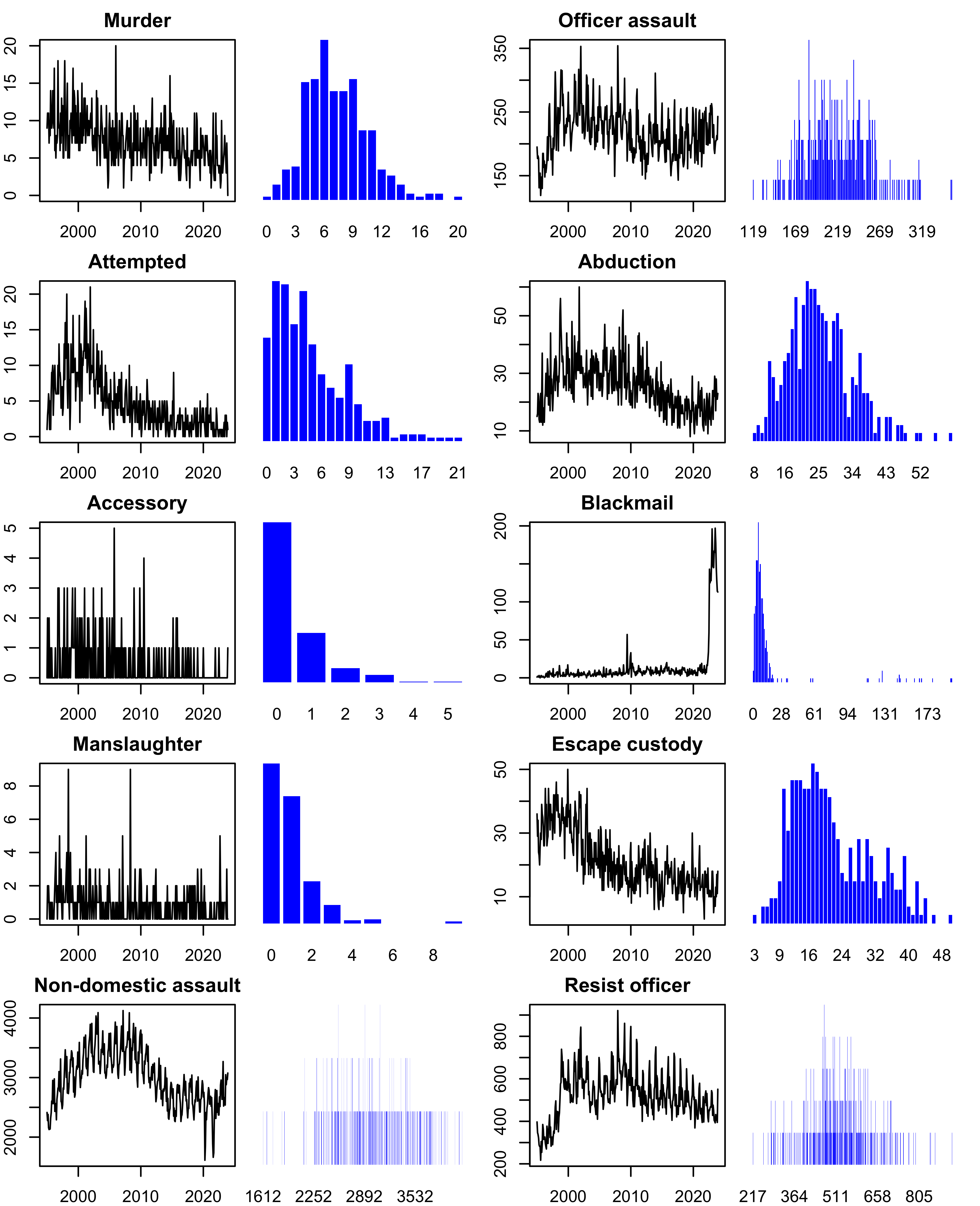}
    \caption{Monthly crime counts in New South Wales, Australia (January 1995 - December 2023).}
    \label{fig:full_crime}
\end{figure}

\begin{figure}
    \centering
    \includegraphics[width=\textwidth]{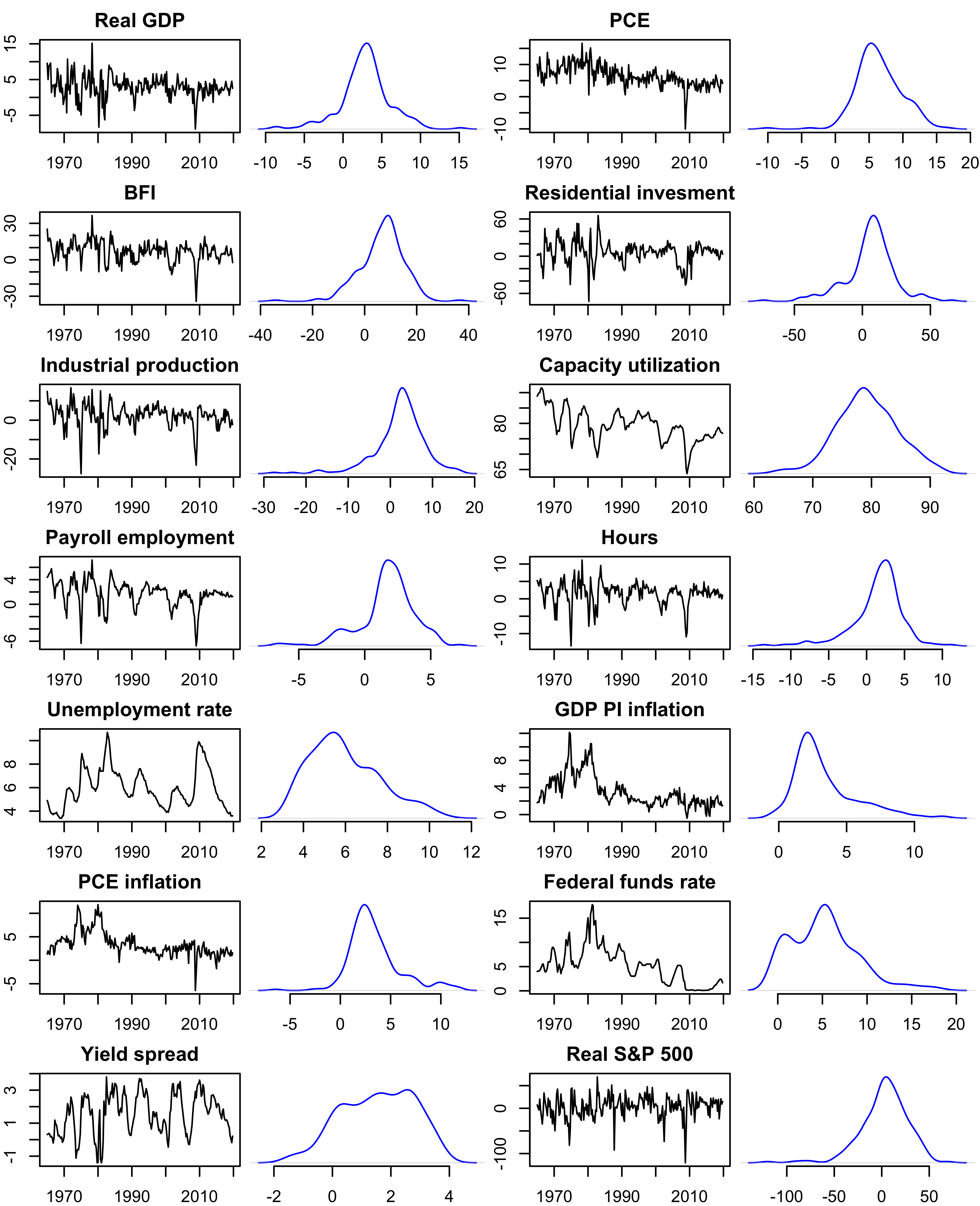}
    \caption{Quarterly US macroeconomic aggregates (1965Q1 - 2019Q4).}
    \label{fig:full_macro}
\end{figure}

\end{document}